\documentclass[11pt,letterpaper]{article}

\usepackage[margin=1in]{geometry}
\usepackage{xcolor}

\usepackage{amsmath,amssymb,amsthm}
\usepackage{aliascnt}
\usepackage{algorithm}
\usepackage{algpseudocode}
\usepackage[hidelinks]{hyperref}
\usepackage[nameinlink,capitalize,noabbrev]{cleveref}

\usepackage{authblk}

\author[1]{Sina Kalantarzadeh}
\author[1]{Kanstantsin Pashkovich}

\affil[1]{Department of Combinatorics and Optimization,
University of Waterloo, Waterloo, Ontario, Canada\\
\texttt{S4kalant@uwaterloo.ca},
\texttt{kpashkovich@uwaterloo.ca}}

\title{Simple and Almost Non-Adaptive \(\frac{1}{2}\)-Approximation for Matroid Prophet Inequalities}

\date{}

\theoremstyle{plain}
\newtheorem{theorem}{Theorem}[section]

\newaliascnt{lemma}{theorem}
\newtheorem{lemma}[lemma]{Lemma}
\aliascntresetthe{lemma}

\newaliascnt{claim}{theorem}

\aliascntresetthe{claim}

\theoremstyle{definition}

\newaliascnt{definition}{theorem}
\newtheorem{definition}[definition]{Definition}
\aliascntresetthe{definition}

\newaliascnt{remark}{theorem}
\newtheorem{remark}[remark]{Remark}
\aliascntresetthe{remark}

\crefname{theorem}{Theorem}{Theorems}
\Crefname{theorem}{Theorem}{Theorems}

\crefname{lemma}{Lemma}{Lemmas}
\Crefname{lemma}{Lemma}{Lemmas}

\crefname{definition}{Definition}{Definitions}
\Crefname{definition}{Definition}{Definitions}

\crefname{claim}{Claim}{Claims}
\Crefname{claim}{Claim}{Claims}

\crefname{remark}{Remark}{Remarks}
\Crefname{remark}{Remark}{Remarks}

\crefname{algorithm}{Algorithm}{Algorithms}
\Crefname{algorithm}{Algorithm}{Algorithms}

\crefname{section}{Section}{Sections}
\Crefname{section}{Section}{Sections}

\begin{document}

\maketitle

\begin{abstract}
Prophet inequalities are a fundamental model for online decision-making under uncertainty. For matroid constraints, Kleinberg and Weinberg gave a tight $\frac{1}{2}$-approximation using adaptive thresholds, while Feldman, Svensson, and Zenklusen obtained a $\frac{1}{4}$-approximation via an online contention resolution scheme (OCRS).

We give the first almost non-adaptive algorithm for general matroid prophet inequalities achieving the optimal $\frac{1}{2}$ guarantee, in fact with respect to the stronger ex-ante relaxation. Starting from an optimal ex-ante solution $x$, we reduce to a Bernoulli instance, replace the original matroid by a stricter direct sum of minors, and assign fixed thresholds to the resulting components. Translating the rule back to the original distributions, an element $e$ can be accepted only when its realized value lies in its top $x_e$-quantile and adding it preserves the corresponding stricter matroid constraint.

We also give a second almost non-adaptive $\frac{1}{2}$-approximation based on a different threshold rule. This formulation extends naturally to intersections of matroids and yields an almost non-adaptive $(q+1)$-approximation for prophet inequalities under the intersection of $q$ arbitrary matroids, again with respect to the ex-ante relaxation. This matches the previously known $(q+1)$ guarantee for intersections of $q$ partition matroids, due to Alon, Pollner, and Weinberg, while extending it to arbitrary matroids.

For the intersection result, each matroid is replaced by a stricter direct sum of minors, and a common surplus vector determines fixed element thresholds across all $q$ constraints. We prove the existence of such a vector using Brouwer's fixed-point theorem and give a polynomial-time procedure to compute it.

\end{abstract}

\newpage

\section{Introduction}
\label{sec:introduction}

Prophet inequalities are a central model for online decision-making under
uncertainty. In the classical prophet inequality, introduced by Krengel and
Sucheston~\cite{KrengelSucheston}, the distributions of a sequence of
independent nonnegative random variables are known in advance, and the
variables are then revealed one by one in an adversarial order to an online
algorithm. Upon seeing the realized value of an item, the algorithm must
immediately and irrevocably decide whether to accept it. In the basic setting,
at most one item may be accepted, which can equivalently be viewed as a
uniform matroid of rank one. The benchmark is a prophet who observes all
realizations in advance and selects the maximum realized value.

A foundational theorem of Krengel and Sucheston shows that the online gambler
can obtain at least half of the prophet's expected value, and this factor is
tight. Samuel-Cahn~\cite{SamuelCahn} gave a particularly simple
threshold-based form of this result: knowing the distributions in advance, one
can compute a single threshold before the process begins, and then accept the
first item whose value exceeds this threshold, regardless of the adversarial
arrival order. Thus, in the single-choice setting, the optimal constant-factor
guarantee can be achieved by a completely non-adaptive rule: all information
used by the threshold rule is fixed before any values are revealed.

A major line of work has generalized prophet inequalities beyond the
single-choice setting. In these variants, the algorithm may accept several
items, subject to a feasibility constraint. Formally, we are given a
downward-closed family $\mathcal F\subseteq 2^U,$ and both the online algorithm and the prophet must select a feasible set in
$\mathcal F$. One of the most important cases is the \emph{matroid prophet
inequality}. Here $\mathcal F=\mathcal I$ is the family of independent sets of
a matroid $M=(U,\mathcal I)$. Each element $e\in U$ has an independent
nonnegative value drawn from a known distribution, and the elements arrive in
an adversarial order. The online algorithm observes the value of each arriving
element and must decide immediately whether to accept it, while maintaining an
independent set with respect to $M$.

Kleinberg and Weinberg~\cite{KleinbergWeinbergJournal} gave a tight
$\frac{1}{2}$-approximation for matroid prophet inequalities. Their result
extends the classical $\frac{1}{2}$ guarantee from the single-choice setting
to arbitrary matroids. Lee and Singla~\cite{LeeSingla2018} later obtained the
same $\frac{1}{2}$ guarantee with respect to the stronger ex-ante relaxation. There is nevertheless an important distinction between the classical
single-choice threshold rule and the Kleinberg--Weinberg algorithm. In the
single-choice setting, one threshold is fixed before the arrival process and
is used until an item is accepted. In the Kleinberg--Weinberg algorithm, by
contrast, the threshold assigned to a future element depends on the current
set of accepted elements and may therefore change whenever this set changes.
Thus, although the thresholds do not react merely to rejected values, they are
not fixed throughout the online process. This state dependence is natural in
the matroid setting, but it motivates the question of whether comparable
guarantees can be obtained using thresholds that are fixed entirely in
advance.

This distinction has motivated the study of \emph{pure non-adaptive} and
\emph{almost non-adaptive} prophet inequalities. Informally, a pure
non-adaptive threshold algorithm fixes all thresholds in advance and continues
to use the original feasibility constraint. An almost non-adaptive, or
constrained non-adaptive, algorithm also fixes its thresholds in advance, but
is allowed to impose an additional, more restrictive feasibility constraint
before the arrival process begins. We formally define these notions, and
separately distinguish whether the arrival order is known in advance, in
Section~\ref{subsec:adaptivity-definitions}.

Fixed thresholds are particularly natural in applications to posted-price
mechanisms: fixed thresholds correspond to fixed posted prices, whereas
state-dependent thresholds correspond to prices that may change after earlier
purchases. For multi-dimensional or constrained-additive buyers, such
state-dependent pricing rules may fail to be truthful, while fixed prices are
substantially easier to interpret and implement.

One way to retain fixed thresholds while obtaining meaningful guarantees is to
strengthen the feasibility constraint before the online process begins.
Chawla et al.~\cite{Chawla2010} introduced such a
``constrained non-adaptive'' approach: their algorithms compute fixed
thresholds in advance, but accept an element only if it passes its threshold
and remains feasible with respect to an additional, more restrictive
constraint. Using this approach, they obtained a
$\frac{1}{3}$-approximation for graphic matroids.

This perspective was later developed systematically by Feldman, Svensson, and
Zenklusen~\cite{FeldmanSvenssonZenklusen} through the framework of online
contention resolution schemes (OCRSs). Their greedy OCRS for matroids gives a
$\frac{1}{4}$-selectable scheme, which yields a
$\frac{1}{4}$-approximation for matroid prophet inequalities with respect to
the ex-ante relaxation. This approach is again almost non-adaptive: the
thresholds are fixed before the arrival process, while an additional
feasibility restriction derived from the fractional solution is used during
the online phase. The OCRS framework is considerably more general and applies
to many constraint families, including intersections of matroids, matchings,
and knapsack constraints. For matroid prophet inequalities, however, its
$\frac{1}{4}$ guarantee loses a factor of $2$ compared with the adaptive
$\frac{1}{2}$-approximation of Kleinberg and Weinberg.

Feldman et al.~\cite{FeldmanSvenssonZenklusen} also show a fundamental
limitation of pure non-adaptive threshold rules for general matroids. They
construct a family of gammoids on $n$ elements for which every pure
non-adaptive threshold algorithm can achieve at most an
\[
O\left(\frac{\log\log n}{\log n}\right)
\]
fraction of the prophet's expected value. Thus, unlike in the single-choice
setting, fixed thresholds together with the original matroid constraint cannot
achieve a constant-factor approximation for general matroids. This
impossibility result is one of the main reasons to distinguish carefully
between pure non-adaptive algorithms and algorithms that impose additional
feasibility constraints in advance.

Chawla et al.~\cite{ChawlaGraphic} subsequently obtained the first purely
non-adaptive constant-factor prophet inequality for graphic matroids. Their
algorithm achieves a $\frac{1}{32}$-approximation. The main difficulty is
that, in a graphic matroid, an edge may be blocked by many other edges, and
fixed thresholds alone do not control the probability that an arriving edge
remains feasible. Their solution exploits the graph structure: by orienting
the edges and considering a random cut, they ensure that each edge is
simultaneously considered and unblocked with constant probability.

Pashkovich and Sayutina~\cite{PashkovichSayutina} further developed the theory
of pure non-adaptive prophet inequalities for restricted classes of matroids.
They obtained constant-factor guarantees for several natural matroid families,
including cographic matroids, $k$-column sparse matroids, regular matroids,
laminar matroids, and truncated partition matroids. In contrast to almost
non-adaptive approaches, these algorithms retain the original matroid
constraint throughout the online process and do not impose an additional
feasibility restriction.

Jiang~\cite{jiang2026costnonadaptivitymatroidprophet} further investigated the
cost of non-adaptivity for laminar and graphic matroids. For laminar matroids,
he gave a pure non-adaptive algorithm achieving a
$\frac{1}{6.311}$-approximation, while showing that no pure non-adaptive
algorithm can achieve better than a
$\frac{1}{2.217}$-approximation. For graphic matroids, he obtained a
$\frac{1}{8}$-approximation and proved that no pure non-adaptive algorithm can
achieve better than a $\frac{1}{3}$-approximation. These results further
narrow the gap between the known achievable and impossible guarantees for
natural classes of matroids.

From this perspective, understanding the precise power of pure non-adaptive
algorithms across different matroid classes remains a central open question,
as does understanding how much additional power is gained by allowing fixed
feasibility restrictions in almost non-adaptive algorithms.

Another natural extension of matroid prophet inequalities is obtained by
requiring the accepted set to be independent in several matroids
simultaneously. For the intersection of $q$ matroids, Kleinberg and
Weinberg~\cite{KleinbergWeinbergJournal} obtained a
$\frac{1}{4q-2}$-approximation. Feldman, Svensson, and
Zenklusen~\cite{FeldmanSvenssonZenklusen} later improved this guarantee to $\frac{1}{e(q+1)}$ using online contention resolution schemes. For the special case of an
intersection of $q$ partition matroids, Correa, Cristi, Fielbaum, Pollner, and
Weinberg~\cite{10.1007/s10107-023-02027-2} obtained a
$\frac{1}{q+1}$-approximation using non-adaptive thresholds. Their
construction assigns auxiliary prices to the different partition-matroid
constraints and defines the threshold of an element by summing the
corresponding prices. The existence of prices satisfying the required balance
conditions is established through Brouwer's fixed-point theorem.

This raises the natural question of whether the $\frac{1}{q+1}$ guarantee can
be extended from partition matroids to intersections of arbitrary matroids
while retaining an almost non-adaptive structure.

\subsection{Our result and techniques}
\label{subsec:our-results}

We give the first almost non-adaptive algorithm for general matroid prophet
inequalities achieving the optimal $\frac{1}{2}$ guarantee, in fact with
respect to the stronger ex-ante relaxation.

Our algorithm starts from the ex-ante relaxation. Given a matroid and value
distributions for its elements, we first solve the ex-ante relaxation and
obtain an optimal solution $x$ in the matroid polytope. As is standard, it
suffices to consider the corresponding Bernoulli instance $(M,x,v)$: each
element $e$ is active with probability $x_e$, and if active has deterministic
value $v_e>0$; otherwise its value is zero. The quantity
\[
\sum_{e\in U}x_ev_e
\]
is the value of the ex-ante relaxation and hence upper bounds the expected
value of the prophet. Our main structural step is to use the vector $x$ to iteratively decompose the
original matroid into \emph{nice} pieces. More precisely, we partition the
ground set into
\[
U_1,\ldots,U_\ell
\]
and construct matroids
\[
M'_1,\ldots,M'_\ell
\]
on these parts such that every set independent in
\[
M'_1\oplus\cdots\oplus M'_\ell
\]
is independent in the original matroid. If
\[
A_{j-1}:=U_1\cup\cdots\cup U_{j-1},
\]
then
\[
M'_j=(M/A_{j-1})|U_j.
\]
The matroids are chosen so that every Bernoulli instance
\[
(M'_j,x|_{U_j},v|_{U_j})
\]
is \emph{nice}. Each component is associated with a fixed threshold $T_j$,
and niceness gives the local guarantee
\[
\mathbb E[v(I\cap U_j)]
\geq
r_{M'_j}(U_j)T_j.
\]

An additional structural property of nice instances is that every positive
Bernoulli value in a component lies above its component threshold.
Consequently, once the decomposition is computed, the Bernoulli online rule
only needs to check whether an element is active and whether accepting it
preserves independence in the corresponding $M'_j$. The thresholds remain
essential in defining and analyzing the decomposition, but no threshold needs
to be recomputed or updated during the arrival process.

Translating the Bernoulli reduction back to the original distributions, an
element $e$ is active when its realized value lies in its fixed top
$x_e$-quantile, with the usual randomization at a cutoff value if necessary.
Thus, the online algorithm on the original distributions uses fixed,
element-specific quantile cutoffs together with fixed stricter feasibility
constraints. The extraction procedure guarantees
\[
\sum_{j=1}^{\ell}r_{M'_j}(U_j)T_j
\geq
\frac{1}{2}\sum_{e\in U}x_ev_e.
\]
Therefore,
\[
\mathbb E[v(I)]
\geq
\frac{1}{2}\sum_{e\in U}x_ev_e.
\]
This gives the optimal $\frac{1}{2}$ guarantee with respect to the ex-ante
relaxation.

Importantly, the preprocessing does not require knowledge of the arrival
order. The decomposition, thresholds, and additional feasibility constraints
depend only on the matroid and the ex-ante Bernoulli instance. Consequently,
the same precomputed algorithm applies without modification to every
adversarial arrival order.

We also give a second almost non-adaptive
$\frac{1}{2}$-approximation for a single matroid, using a different threshold
rule on the Bernoulli instance. As before, this yields fixed component
thresholds together with a stricter direct-sum matroid. This alternative
formulation gives the same $\frac{1}{2}$ guarantee, while having the
additional advantage that it extends naturally to intersections of matroids.

For a loopless Bernoulli instance $(M,x,v)$ and a nonempty set $S\subseteq U$,
we define its \emph{surplus threshold} $T'(S)$ as the unique nonnegative value
satisfying
\[
T'(S)r_M(S)
=
\sum_{e\in S}x_e\bigl(v_e-T'(S)\bigr)^+.
\]
We call an instance \emph{surplus-nice} if
\[
T'(S)\leq T'(U)
\]
for every nonempty set $S\subseteq U$. A surplus-nice instance can again be
handled using a single fixed threshold, yielding expected value at least
\[
r_M(U)T'(U).
\]

We then use an analogous extract-and-contract procedure, repeatedly extracting
a set of maximum surplus threshold. A maximum-surplus-threshold set can always
be chosen to be a flat. Consequently, contracting an extracted piece does not
create loops, and all residual matroids remain loopless. The resulting pieces
are surplus-nice, the corresponding thresholds are non-increasing, and the
extracted direct sum is more restrictive than the original matroid. Combining
these properties again gives
\[
\mathbb E[\mathrm{ALG}]
\geq
\frac{1}{2}\sum_{e\in U}x_ev_e.
\]

The surplus-threshold viewpoint extends to the intersection of several
matroids. Let
\[
M_1,\ldots,M_q
\]
be matroids on a common ground set $U$, and let
\[
x\in\bigcap_{i=1}^qP_{M_i}.
\]
We obtain an almost non-adaptive
$\frac{1}{q+1}$-approximation with respect to the ex-ante relaxation.

The construction follows the same extract-and-contract phenomenon as the
single-matroid surplus-threshold algorithm. The main difference is that there
are now $q$ matroids and therefore $q$ decompositions that must be coupled. To
achieve this, we introduce a common nonnegative surplus vector
\[
y\in\mathbb R_+^U.
\]
For each matroid $M_i$, we repeatedly extract a maximum-density set with
respect to $y$, restrict to this set, and contract it. Each extracted piece can
again be chosen to be a flat in the corresponding residual matroid. The
resulting pieces form a stricter direct-sum matroid whose independent sets are
independent in $M_i$.

Each extracted piece of $M_i$ is assigned a price equal to its $y$-density.
If an element $e$ belongs to a piece with price $p_e^i$, then $p_e^i$ is the
threshold contribution of matroid $M_i$ to element $e$. We define the final
threshold by
\[
T_e:=\sum_{i=1}^qp_e^i.
\]
The surplus vector is chosen so that
\[
y_e=x_e(v_e-T_e)^+
\]
for every $e\in U$. Hence the common vector $y$ determines the decompositions
and their prices, while the prices determine the thresholds, which in turn
determine $y$. We establish the existence of such a vector using Brouwer's
fixed-point theorem. In the appendix, we further show that a suitable vector
$y$ can be computed in polynomial time.

The online algorithm maintains feasibility in all $q$ fixed stricter
direct-sum matroids and uses the fixed thresholds $T_e$. All thresholds and
all additional feasibility constraints are determined before the arrival
process begins, and none of them requires knowledge of the arrival
permutation. Hence the algorithm is both almost non-adaptive and
order-oblivious. The analysis shows
\[
\mathbb E[\mathrm{ALG}]
\geq
y(U)
\]
and
\[
\sum_{e\in U}x_ev_e
\leq
(q+1)y(U).
\]
Consequently,
\[
\mathbb E[\mathrm{ALG}]
\geq
\frac{1}{q+1}\sum_{e\in U}x_ev_e.
\]

It is natural to ask how close the linear dependence on $q$ is to being
optimal. Given a prime number $p$, Kleinberg and
Weinberg~\cite{KleinbergWeinbergJournal} gave a downward-closed feasibility
system for which the ratio between the prophet and the optimal online
algorithm is $\Theta(p)$, while the same feasibility system can be represented
as the intersection of $p^2$ partition matroids. This gives an
$\Omega(\sqrt q)$ lower bound on the approximation factor for intersections
of $q$ matroids.

Saxena, Velusamy, and
Weinberg~\cite{saxena_et_al:LIPIcs.ITCS.2023.95} later strengthened this
lower bound by showing that the same underlying hard feasibility system can be
represented as the intersection of asymptotically fewer partition matroids.
Their result implies that no online algorithm can guarantee more than
\[
q^{-1/2-\Omega(1/\log\log q)}
\]
of the prophet's expected value for all intersections of $q$ matroids, even
when all the matroids are partition matroids. Thus, together with our
$\frac{1}{q+1}$ guarantee, a substantial gap remains. An interesting open
question is whether the Kleinberg--Weinberg hard feasibility system can be
represented as the intersection of still fewer matroids.

Table~\ref{tab:adaptivity-landscape} summarizes the known achievable and
impossible guarantees most relevant to this work. The three algorithmic
columns distinguish adaptive thresholds using the original feasibility
constraint, pure non-adaptive thresholds using the original feasibility
constraint, and almost non-adaptive thresholds with an additional fixed
feasibility restriction. All guarantees are written as fractions of the
prophet benchmark. When an achievable guarantee holds with respect to the
stronger ex-ante relaxation, we indicate this explicitly.

\begin{table}[ht]
\centering
\scriptsize
\renewcommand{\arraystretch}{1.5}
\begin{tabular}{
|p{0.14\textwidth}|
p{0.25\textwidth}|
p{0.28\textwidth}|
p{0.28\textwidth}|}
\hline
\textbf{Constraint}
&
\textbf{Adaptive thresholds}
&
\textbf{Pure non-adaptive}
&
\textbf{Almost non-adaptive}
\\
\hline

General matroid
&
\textbf{Achievable:}
$\frac{1}{2}$~\cite{KleinbergWeinbergJournal};
$\frac{1}{2}$ ex-ante~\cite{LeeSingla2018}.
\newline
\textbf{Impossibility:}
no better than $\frac{1}{2}$.
&
\textbf{Achievable:}
no nontrivial general guarantee is known.
\newline
\textbf{Impossibility:}
at most
$O\!\left(\frac{\log\log n}{\log n}\right)$
for some gammoids~\cite{FeldmanSvenssonZenklusen}.
&
\textbf{Achievable:}
$\frac{1}{2}$ ex-ante (this work);
previously $\frac{1}{4}$ ex-ante
~\cite{FeldmanSvenssonZenklusen}.
\newline
\textbf{Impossibility:}
no better than $\frac{1}{2}$.
\\
\hline

Graphic matroid
&
\textbf{Achievable:}
$\frac{1}{2}$~\cite{KleinbergWeinbergJournal}.
\newline
\textbf{Impossibility:}
no better than $\frac{1}{2}$.
&
\textbf{Achievable:}
$\frac{1}{8}$
~\cite{jiang2026costnonadaptivitymatroidprophet}.
\newline
\textbf{Impossibility:}
no better than $\frac{1}{3}$
~\cite{jiang2026costnonadaptivitymatroidprophet}.
&
\textbf{Achievable:}
$\frac{1}{2}$ ex-ante (this work);
previously $\frac{1}{3}$~\cite{Chawla2010}
and $\frac{1}{4}$ ex-ante
~\cite{FeldmanSvenssonZenklusen}.
\newline
\textbf{Impossibility:}
no better than $\frac{1}{2}$.
\\
\hline

Laminar matroid
&
\textbf{Achievable:}
$\frac{1}{2}$~\cite{KleinbergWeinbergJournal}.
\newline
\textbf{Impossibility:}
no better than $\frac{1}{2}$.
&
\textbf{Achievable:}
$\frac{1}{6.311}$
~\cite{jiang2026costnonadaptivitymatroidprophet}.
\newline
\textbf{Impossibility:}
no better than $\frac{1}{2.217}$
~\cite{jiang2026costnonadaptivitymatroidprophet}.
&
\textbf{Achievable:}
$\frac{1}{2}$ ex-ante (this work);
previously $\frac{1}{4}$ ex-ante
~\cite{FeldmanSvenssonZenklusen}.
\newline
\textbf{Impossibility:}
no better than $\frac{1}{2}$.
\\
\hline

Intersection of $q$ matroids
&
\textbf{Achievable:}
$\frac{1}{4q-2}$
~\cite{KleinbergWeinbergJournal}.
\newline
\textbf{Impossibility:}
at most
$q^{-1/2-\Omega(1/\log\log q)}$
~\cite{saxena_et_al:LIPIcs.ITCS.2023.95}.
&
\textbf{Achievable:}
$\frac{1}{q+1}$ for intersections of $q$ partition matroids
~\cite{10.1007/s10107-023-02027-2};
no corresponding general guarantee is known for arbitrary matroids.
\newline
\textbf{Impossibility:}
at most
$q^{-1/2-\Omega(1/\log\log q)}$,
even for partition matroids
~\cite{saxena_et_al:LIPIcs.ITCS.2023.95}.
&
\textbf{Achievable:}
$\frac{1}{q+1}$ ex-ante (this work);
previously
$\frac{1}{e(q+1)}$
~\cite{FeldmanSvenssonZenklusen}.
\newline
\textbf{Impossibility:}
at most
$q^{-1/2-\Omega(1/\log\log q)}$
~\cite{saxena_et_al:LIPIcs.ITCS.2023.95}.
\\
\hline
\end{tabular}

\caption{Selected achievable guarantees and impossibility results for matroid
prophet inequalities under different notions of adaptivity. All quantities are
written as fractions of the prophet benchmark. An impossibility entry of $c$
means that there are instances on which no algorithm in the corresponding
class can guarantee more than a $c$ fraction of the prophet value. Guarantees
marked ``ex-ante'' hold with respect to the stronger ex-ante relaxation.}
\label{tab:adaptivity-landscape}
\end{table}

\subsection{Comparison with related work and computational aspects}
\label{subsec:comparison-related-work}

Our first algorithm achieves the optimal $\frac{1}{2}$ guarantee for a
general matroid, but its information structure differs from the known
$\frac{1}{2}$ algorithms. In the Kleinberg--Weinberg algorithm, the threshold
used for a future element may change when the set of accepted elements
changes. By contrast, all thresholds and all additional feasibility
constraints in our algorithm are determined before any values are observed
and never change during the online phase.

Lee and Singla also obtain the optimal $\frac{1}{2}$ guarantee against the
ex-ante relaxation under adversarial arrival. Their adversarial-order OCRS
formulation, however, assumes that the arrival permutation is known to the
algorithm in advance. Our preprocessing does not require this information:
the same decomposition, thresholds, and online rule apply to every
adversarial arrival order.

Our algorithm is not purely non-adaptive, because it replaces the original
feasibility constraint by a more restrictive one before the online phase.
More precisely, it constructs matroids
\[
M'_1,\ldots,M'_\ell
\]
such that every set independent in
\[
M'_1\oplus\cdots\oplus M'_\ell
\]
is independent in the original matroid. Thus, our result is almost
non-adaptive in the sense formalized in
Section~\ref{subsec:adaptivity-definitions}: the thresholds and the additional
feasibility restriction are fixed in advance. The impossibility result of
Feldman et al. for pure non-adaptive thresholds shows that allowing such an
additional feasibility restriction fundamentally changes what is possible for
general matroids.

Our matroid-intersection result should also be compared with the previous
bounds for intersections of matroids. Kleinberg and Weinberg obtained a
$\frac{1}{4q-2}$-approximation using adaptive thresholds, and Feldman,
Svensson, and Zenklusen obtained a $\frac{1}{e(q+1)}$ guarantee using OCRSs. For the special case of intersections of $q$ partition
matroids, Correa, Cristi, Fielbaum, Pollner, and
Weinberg~\cite{10.1007/s10107-023-02027-2} obtained a
$\frac{1}{q+1}$ guarantee using non-adaptive thresholds. Our result extends
the same $\frac{1}{q+1}$ approximation factor to intersections of $q$
arbitrary matroids.

There is also a conceptual similarity between our construction and the
non-adaptive partition-matroid-intersection algorithm of Correa, Cristi,
Fielbaum, Pollner, and Weinberg. In both cases, the threshold of an element is
obtained by summing contributions from the individual matroid constraints, and
a fixed-point argument is used to choose these contributions consistently. In
our setting, however, the contribution of each matroid is generated by an
extract-and-contract density decomposition with respect to a common surplus
vector. This allows the construction to work for arbitrary matroids rather
than only partition matroids. The price of this generality is that we impose
additional feasibility constraints in advance, and hence our algorithm is
almost non-adaptive rather than purely non-adaptive.

All of our algorithms can be implemented in polynomial time. For the
single-matroid algorithms, once the ex-ante solution is computed, the
decompositions and thresholds can be obtained in polynomial time using an
independence oracle. For the matroid-intersection algorithm, the additional
ingredient is the coupled surplus vector $y$. Although the main body proves
the existence of such a vector using Brouwer's fixed-point theorem, the
appendix gives a polynomial-time method for computing a suitable vector via a
convex optimization formulation and subgradient methods.

\subsection{Organization of the paper}
\label{subsec:organization}

\Cref{sec:preliminaries} introduces the online model, the notions of adaptivity and arrival
order used throughout the paper, the ex-ante relaxation and Bernoulli
reduction, and the matroid preliminaries needed later. \Cref{sec:first-algorithm}
presents our first almost non-adaptive $\frac{1}{2}$-approximation for a
single matroid. \Cref{sec:surplus-algorithm} develops the alternative
surplus-threshold formulation and proves a second almost non-adaptive
$\frac{1}{2}$-approximation. \Cref{sec:q-intersection} extends this framework
to intersections of $q$ arbitrary matroids and proves the $\frac{1}{q+1}$
guarantee. The appendix contains the proof of the ex-ante reduction and the
polynomial-time implementations of the extraction, density decomposition, and
coupled fixed-point procedures.

\section{Preliminaries}
\label{sec:preliminaries}

\subsection{Online Model, Arrival Order, and Notions of Adaptivity}
\label{subsec:adaptivity-definitions}

We first formalize the online model and the notions of adaptivity used
throughout the paper. Let $U$ be a finite ground set and let $\mathcal F\subseteq 2^U$ be a downward-closed feasibility family. Each element $e\in U$ has an
independent nonnegative random value $X_e$ drawn from a known distribution.

\begin{definition}[Adversarial arrival order]
An \emph{adversarial arrival order} is an arbitrary permutation
\[
\pi=(e_1,\ldots,e_n)
\]
of $U$ chosen before the random values are realized. The adversary may choose
$\pi$ using the feasibility constraint and the known value distributions, but
not using the realizations of the random variables or the internal randomness
subsequently used by the online algorithm. The performance guarantee of an
algorithm is required to hold for every such permutation.
\end{definition}

We separately distinguish whether the permutation is revealed to the
algorithm before the online process begins.

\begin{definition}[Order-aware and order-oblivious algorithms]
An online algorithm is \emph{order-aware} if its preprocessing is allowed to
depend on the adversarial arrival permutation $\pi$. It is
\emph{order-oblivious} if its preprocessing and its online decision rule do
not require advance knowledge of $\pi$.
\end{definition}

An order-oblivious algorithm of course observes the identity of an element
when it arrives; the definition only means that the future arrival permutation
is not known in advance.

The amount of information available about the arrival order is separate from
the adaptivity of the online decision rule.

\begin{definition}[Adaptive online algorithm]
An online algorithm is \emph{adaptive} if its decision rule for an arriving
element may depend on the history observed so far, including the identities
and realized values of previously arrived elements and the set of elements
accepted so far. In particular, an adaptive threshold algorithm may change the
threshold assigned to a future element as a function of this history.
\end{definition}

When discussing adaptive threshold algorithms in
Table~\ref{tab:adaptivity-landscape}, we use the original feasibility family
$\mathcal F$ throughout the online process. Thus, the distinction from the two
notions below concerns both whether thresholds are fixed in advance and
whether an additional feasibility restriction is imposed.

\begin{definition}[Pure non-adaptive threshold algorithm]
A \emph{pure non-adaptive threshold algorithm} fixes, before the arrival
process begins, a threshold $\tau_e$ for every element $e\in U$. When $e$
arrives, the algorithm accepts it if
\[
X_e\geq\tau_e
\]
and
\[
I\cup\{e\}\in\mathcal F,
\]
where $I$ is the set accepted so far. The thresholds remain unchanged
throughout the online process, and feasibility is always tested with respect
to the original family $\mathcal F$.
\end{definition}

\begin{definition}[Almost non-adaptive threshold algorithm]
An \emph{almost non-adaptive}, or \emph{constrained non-adaptive}, threshold
algorithm fixes before the arrival process begins both thresholds
\[
\{\tau_e\}_{e\in U}
\]
and a downward-closed feasibility family
\[
\mathcal F'\subseteq\mathcal F.
\]
When $e$ arrives, the algorithm accepts it if
\[
X_e\geq\tau_e
\]
and
\[
I\cup\{e\}\in\mathcal F'.
\]
Thus, both the thresholds and the additional feasibility restriction are fixed
before the online process begins.
\end{definition}

As usual, fixed randomization may be used at a threshold equality. The
preprocessing in the non-adaptive models may depend on the known
distributions, the feasibility constraint, and randomness sampled before the
online process. Whether it may additionally depend on the arrival permutation
is captured separately by the notions of order-awareness and
order-obliviousness.

All algorithms developed in this paper are order-oblivious: their
preprocessing does not use the arrival permutation, and their guarantees hold
for every adversarial arrival order as defined above.

We express approximation guarantees as fractions of the benchmark.

\begin{definition}[Approximation guarantee]
For $c\in[0,1]$, an online algorithm is a \emph{$c$-approximation} if, for
every instance and every adversarial arrival order,
\[
\mathbb E[\mathrm{ALG}]
\geq
c\,\mathbb E[\mathrm{PROPHET}],
\]
where $\mathrm{PROPHET}$ denotes the value of the optimal feasible solution
after all values are revealed.
\end{definition}

When the benchmark is the ex-ante relaxation rather than the prophet
benchmark, we state this explicitly.

\subsection{The Matroid Prophet Inequality Problem}
\label{subsec:matroid-prophet}

We next specialize the general online model to matroid constraints.

\begin{definition}[Matroid prophet inequality]
An instance of the matroid prophet inequality problem consists of a matroid
\[
M=(U,\mathcal I)
\]
on a finite ground set $U$, together with independent nonnegative random
variables $\{X_e\}_{e\in U}$. The elements arrive in an adversarial order.
When an element $e$ arrives, the online algorithm observes the realized value
of $X_e$ and must immediately and irrevocably decide whether to accept $e$.
The set of accepted elements must remain independent in $M$.

The benchmark is the expected value of a prophet who observes all
realizations in advance and selects a maximum-value independent set:
\[
\mathbb E\left[
\max_{I\in\mathcal I}
\sum_{e\in I}X_e
\right].
\]
\end{definition}

Throughout the paper, we assume that each matroid is given by an independence
oracle. Given a set $S\subseteq U$, the oracle determines whether
$S\in\mathcal I$.

We will also consider feasibility given by the intersection of several
matroids.

\begin{definition}[Matroid-intersection prophet inequality]
An instance of the matroid-intersection prophet inequality problem consists of
matroids
\[
M_i=(U,\mathcal I_i),
\qquad i=1,\ldots,q,
\]
on a common finite ground set $U$, together with independent nonnegative
random variables $\{X_e\}_{e\in U}$. The elements arrive in an adversarial
order, and upon observing the realized value of an arriving element, the
online algorithm must immediately and irrevocably decide whether to accept it.
The accepted set must remain independent in every $M_i$. The benchmark is
\[
\mathbb E\left[
\max_{I\in\bigcap_{i=1}^q\mathcal I_i}
\sum_{e\in I}X_e
\right].
\]
\end{definition}

\subsection{The Matroid Polytope and the Ex-Ante Relaxation}
\label{subsec:ex-ante}

For a matroid $M=(U,\mathcal I)$, the matroid polytope is
\[
P_M
:=
\operatorname{conv}\{\mathbf 1_I:I\in\mathcal I\}
\subseteq\mathbb R^U.
\]
Equivalently,
\[
P_M
=
\left\{
x\in\mathbb R_{\geq0}^U:
\sum_{e\in S}x_e\leq r_M(S)
\text{ for every }S\subseteq U
\right\},
\]
where $r_M$ denotes the rank function of $M$.

The ex-ante relaxation is obtained by relaxing the prophet's choice to a
vector of marginal selection probabilities. For each element $e$, define
$R_e(z)$ to be the maximum expected contribution obtainable from $e$ by
selecting it with probability at most $z$. Equivalently,
\[
R_e(z)
=
\max\left\{
\mathbb E[X_ez_e(X_e)]:
0\leq z_e(X_e)\leq1,\;
\mathbb E[z_e(X_e)]\leq z
\right\},
\]
where $z_e(X_e)$ is a randomized selection rule depending only on the realized
value of $X_e$.

An optimal rule selects the largest realized values of $X_e$, with possible
randomization at one cutoff value. The ex-ante relaxation for a single matroid
is
\[
\max\left\{
\sum_{e\in U}R_e(x_e):
x\in P_M
\right\}.
\]
This is an upper bound on the prophet benchmark. Indeed, the vector of
marginal selection probabilities of the prophet belongs to $P_M$, and for
each element the function $R_e$ gives the largest expected contribution
achievable with the corresponding marginal probability.

For the intersection of matroids
\[
M_1,\ldots,M_q,
\]
we use
\[
P_{\cap}
:=
\bigcap_{i=1}^qP_{M_i}.
\]
If a random set is independent in every $M_i$, then its vector of marginal
selection probabilities belongs to $P_{M_i}$ for every $i$, and hence belongs
to $P_{\cap}$. Therefore,
\[
\max\left\{
\sum_{e\in U}R_e(x_e):
x\in\bigcap_{i=1}^qP_{M_i}
\right\}
\]
is an upper bound on the prophet benchmark for the intersection of the
$q$ matroids.

\subsection{Bernoulli Instances and the Ex-Ante Reduction}
\label{subsec:bernoulli-reduction}

We use the standard ex-ante reduction to Bernoulli
instances~\cite{ChawlaGraphic,FeldmanSvenssonZenklusen}. A Bernoulli instance
is specified by a pair $(x,v)$, where
\[
x\in[0,1]^U,
\qquad
v\in\mathbb R_{\geq0}^U.
\]
Element $e$ has value
\[
X_e=
\begin{cases}
v_e,&\text{with probability }x_e,\\
0,&\text{with probability }1-x_e.
\end{cases}
\]
We say that $e$ is \emph{active} if $X_e=v_e$. Thus, the elements are
independently active, and an active element has deterministic value $v_e$.

The feasibility constraint will be clear from context. For a single matroid,
we write the Bernoulli instance as
\[
(M,x,v),
\]
and for an intersection of matroids $M_1,\ldots,M_q$, we write
\[
(M_1,\ldots,M_q,x,v).
\]

In general, a Bernoulli instance need not satisfy any polyhedral constraint on
$x$. However, when it arises from the ex-ante relaxation, $x$ belongs to the
corresponding relaxation. For a single matroid,
\[
x\in P_M,
\]
and for the intersection of $q$ matroids,
\[
x\in\bigcap_{i=1}^qP_{M_i}.
\]

We now state the ex-ante reduction in a form that applies to both settings.
The proof is deferred to Appendix~\ref{app:exante-reduction}.

\begin{lemma}[Ex-ante reduction to Bernoulli instances]
\label{lem:ex-ante-bernoulli}
Let $\mathcal P\subseteq[0,1]^U$ be a relaxation such that the vector of
marginal selection probabilities of every feasible random set belongs to
$\mathcal P$. Suppose that for some $c\in[0,1]$, every Bernoulli instance
$(x,v)$ with $x\in\mathcal P$ admits an online algorithm with expected value at
least
\[
c\sum_{e\in U}x_ev_e.
\]
Then every corresponding prophet inequality instance admits an online
algorithm with expected value at least $c$ times the ex-ante relaxation value
\[
\max\left\{
\sum_{e\in U}R_e(x_e):
x\in\mathcal P
\right\}.
\]
In particular, the algorithm obtains at least a $c$ fraction of the prophet
benchmark. Moreover, the reduction preserves pure non-adaptivity, almost
non-adaptivity, and order-obliviousness.
\end{lemma}

For a single matroid, we apply
\Cref{lem:ex-ante-bernoulli} with
\[
\mathcal P=P_M.
\]
For the intersection of $q$ matroids, we apply it with
\[
\mathcal P=\bigcap_{i=1}^qP_{M_i}.
\]

\subsection{Matroid Operations}
\label{subsec:matroid-operations}

We use standard operations on matroids. Let $M=(U,\mathcal I)$ be a matroid with rank function $r_M$. For a set $A\subseteq U$, the \emph{restriction} of $M$ to $A$, denoted
$M|A$, is the matroid on ground set $A$ whose independent sets are
\[
\mathcal I(M|A)
:=
\{I\subseteq A:I\in\mathcal I\}.
\]
Its rank function satisfies
\[
r_{M|A}(S)=r_M(S)
\qquad
\text{for every }S\subseteq A.
\]

For a set $A\subseteq U$, the \emph{contraction} of $A$, denoted $M/A$, is
the matroid on ground set $U\setminus A$ with rank function
\[
r_{M/A}(S)
=
r_M(S\cup A)-r_M(A)
\qquad
\text{for every }S\subseteq U\setminus A.
\]
We repeatedly use
\[
r_M(A\cup S)
=
r_M(A)+r_{M/A}(S),
\qquad
S\subseteq U\setminus A.
\]

The \emph{closure} of a set $A\subseteq U$, denoted
$\operatorname{cl}_M(A)$, is
\[
\operatorname{cl}_M(A)
:=
\{e\in U:r_M(A\cup\{e\})=r_M(A)\}.
\]
A set $A\subseteq U$ is a \emph{flat} if
\[
\operatorname{cl}_M(A)=A.
\]

We also use direct sums. Suppose
\[
M_1=(U_1,\mathcal I_1),\ldots,M_\ell=(U_\ell,\mathcal I_\ell)
\]
are matroids on pairwise disjoint ground sets. Their \emph{direct sum}
\[
M_1\oplus\cdots\oplus M_\ell
\]
is the matroid on $U_1\cup\cdots\cup U_\ell$ whose independent sets are
\[
\mathcal I(M_1\oplus\cdots\oplus M_\ell)
:=
\left\{
I\subseteq U_1\cup\cdots\cup U_\ell:
I\cap U_j\in\mathcal I_j
\text{ for every }j
\right\}.
\]
Its rank function is additive:
\[
r_{M_1\oplus\cdots\oplus M_\ell}(S)
=
\sum_{j=1}^{\ell}r_{M_j}(S\cap U_j).
\]

In the algorithms of the following sections, we repeatedly construct a direct
sum
\[
M':=M'_1\oplus\cdots\oplus M'_\ell
\]
whose independent sets form a subfamily of the independent sets of the input
matroid. We use the following structural lemma repeatedly.

\begin{lemma}
\label{lem:direct-sum-submatroid}
Let $M=(U,\mathcal I)$ be a matroid, and let
\[
U_1,\ldots,U_\ell
\]
be pairwise disjoint subsets of $U$. Define
\[
A_i:=U_1\cup\cdots\cup U_i,
\qquad
A_0:=\emptyset,
\]
and, for every $i=1,\ldots,\ell$, let
\[
M'_i:=(M/A_{i-1})|U_i.
\]
Then every set independent in
\[
M'_1\oplus\cdots\oplus M'_\ell
\]
is independent in $M$.
\end{lemma}

\begin{proof}
Let
\[
I\in\mathcal I(M'_1\oplus\cdots\oplus M'_\ell),
\]
and write
\[
I_i:=I\cap U_i.
\]
Thus, $I_i$ is independent in $M'_i$ for every $i$. We prove by induction on
$k$ that
\[
I_1\cup\cdots\cup I_k
\]
is independent in $M$. The claim is immediate for $k=1$, since
\[
M'_1=M|U_1.
\]
Now suppose
\[
I_1\cup\cdots\cup I_{k-1}
\]
is independent in $M$. Extend this set to a basis $B$ of $A_{k-1}$ in $M$.
Since
\[
I_k\in\mathcal I\bigl((M/A_{k-1})|U_k\bigr),
\]
the set $I_k$ is independent in $M/A_{k-1}$. Hence
\[
B\cup I_k
\]
is independent in $M$. Since
\[
I_1\cup\cdots\cup I_k
\subseteq
B\cup I_k,
\]
downward closure implies that
\[
I_1\cup\cdots\cup I_k
\]
is independent in $M$. Taking $k=\ell$ proves the claim.
\end{proof}

\section{An Almost Non-Adaptive $\frac{1}{2}$-Approximation for Matroid Prophet Inequalities}
\label{sec:first-algorithm}

By the reduction in the previous section, given an instance of the matroid
prophet inequality, it suffices to solve the ex-ante relaxation and consider the
resulting Bernoulli instance $(M,x,v)$ with $x\in P_M$. In this section, we
present our first almost non-adaptive $\frac{1}{2}$-approximation for the
single-matroid setting. The algorithm uses thresholds fixed before the arrival
process begins, after replacing the original matroid by a more restrictive
direct-sum matroid.

We first identify a sufficient condition under which a single fixed threshold
already gives a strong guarantee. This motivates the following definition,
which plays a central role in the construction.

\begin{definition}[Threshold of a set]
Let \((M,x,v)\) be a Bernoulli instance with \(M=(U,\mathcal I)\), rank function
\(r\), and \(x_e>0\) for every \(e\in U\). For a set \(S\subseteq U\), define
\[
    x(S):=\sum_{e\in S}x_e,
    \qquad
    w(S):=\sum_{e\in S}x_ev_e.
\]
For every nonempty set \(S\subseteq U\), define
\[
    T(S):=\frac{w(S)}{r(S)+x(S)}.
\]
\end{definition}

\begin{definition}[Nice Bernoulli matroid instance]
A Bernoulli instance \((M,x,v)\), with \(x_e>0\) for every \(e\in U\), is called
\emph{nice} if
\[
    T(S)\le T(U)
    \qquad
    \text{for every nonempty }S\subseteq U.
\]
We emphasize that this definition does not require \(x\in P_M\).
\end{definition}

The next lemma shows that niceness is a sufficient condition for a single
threshold to obtain the desired local guarantee.

\begin{lemma}
\label{lem:nice-values-above-threshold}
Let $(M,x,v)$ be a nice Bernoulli matroid instance and let
\[
T:=T(U).
\]
Then $v_e\geq T$ for every $e\in U$.
\end{lemma}

\begin{proof}
Suppose, for contradiction, that there exists $e\in U$ such that $v_e<T.$
Let $U':=U\setminus\{e\}.$ First, $U'$ is nonempty. Indeed, if $U=\{e\}$, then
\[
T
=
\frac{x_ev_e}{r_M(U)+x_e}
\leq v_e,
\]
contradicting $v_e<T$. Since $x_e>0$ and $v_e<T$,
\[
\begin{aligned}
w(U')
&=w(U)-x_ev_e\\
&>w(U)-x_eT\\
&=T\bigl(r_M(U)+x(U)\bigr)-x_eT\\
&=T\bigl(r_M(U)+x(U')\bigr).
\end{aligned}
\]
Since
\[
r_M(U')\leq r_M(U),
\]
it follows that
\[
w(U')
>
T\bigl(r_M(U')+x(U')\bigr).
\]
Therefore,
\[
T(U')
=
\frac{w(U')}{r_M(U')+x(U')}
>
T=T(U),
\]
contradicting niceness.
\end{proof}

\begin{lemma}[Single-threshold guarantee for nice matroids]
\label{lem:nice-matroid-single-threshold}
Let \((M,x,v)\) be a nice Bernoulli matroid instance. Let
\[
    T:=T(U)=\frac{w(U)}{r(U)+x(U)}.
\]
Consider the single-threshold algorithm that accepts an active element \(e\) if
\(v_e\ge T\) and accepting \(e\) preserves independence in \(M\). If \(I\)
denotes the random independent set selected by the algorithm, then
\[
    \mathbb E[v(I)]
    \geq
    r(U)T.
\]
\end{lemma}

\begin{proof}
Let $T:=T(U).$ By \Cref{lem:nice-values-above-threshold}, $v_e\geq T$ for every $e\in U$. For every set $S\subseteq U$, define
\[
p(S):=\sum_{e\in S}x_e(v_e-T)
      =w(S)-Tx(S).
\]
Since the instance is nice, for every nonempty set $S\subseteq U$, $T(S)\leq T.$ Equivalently,
\[
w(S)\leq T\bigl(r(S)+x(S)\bigr),
\]
and hence
\begin{equation}
\label{eq:nice-density-bound}
    p(S)=w(S)-Tx(S)\leq Tr(S).
\end{equation}
Inequality~\eqref{eq:nice-density-bound} is also trivial for $S=\emptyset$. For the whole ground set $U$, by the definition of $T=T(U)$,
\[
p(U)
=
w(U)-Tx(U)
=
Tr(U).
\]

Let $I$ be the random independent set selected by the algorithm, and let
\[
F:=\operatorname{cl}_M(I)
\]
be its closure. Let $I'$ denote the set of active elements. Since every
element satisfies $v_e\geq T$, the algorithm greedily constructs a maximal
independent subset of $I'$. Therefore,
\[
\operatorname{cl}_M(I)=\operatorname{cl}_M(I'),
\]
and hence $F=\operatorname{cl}_M(I').$ Since $I$ is independent and spans $F$,
\[
r(F)=|I|.
\]

Fix an element $e\in U$. If $e$ has no space when it arrives, then it is
spanned by elements that have already been accepted. Since these elements are
contained in $I$, it follows that $e\in F$. Hence
\[
\Pr[e\text{ has space}]
\geq
1-\Pr[e\in F].
\]
The arrival order is fixed before the Bernoulli realizations. Therefore, the
event that $e$ has space depends only on the activations of elements preceding
$e$ in this fixed order and is independent of the activation of $e$. Thus,
\[
\begin{aligned}
\mathbb E\left[\sum_{e\in I}(v_e-T)\right]
&=
\sum_{e\in U}x_e(v_e-T)\Pr[e\text{ has space}]\\
&\geq
\sum_{e\in U}x_e(v_e-T)\bigl(1-\Pr[e\in F]\bigr)\\
&=
p(U)-\mathbb E[p(F)].
\end{aligned}
\]

Consequently,
\[
\begin{aligned}
\mathbb E[v(I)]
&=
T\mathbb E[|I|]
+
\mathbb E\left[\sum_{e\in I}(v_e-T)\right]\\
&\geq
p(U)+\mathbb E\left[Tr(F)-p(F)\right]\\
&\geq
p(U)
=
Tr(U),
\end{aligned}
\]
where the second inequality follows from
\eqref{eq:nice-density-bound} applied to the random set $F$.
\end{proof}

\subsection{The Extraction Procedure}
\label{subsec:first-extraction}

We now describe the preprocessing step. Let \((M,x,v)\) be the Bernoulli
instance obtained from the ex-ante relaxation, so \(x\in P_M\). We first remove
all elements \(e\) with \(x_e=0\), and hence assume throughout this subsection
that
\[
    x_e>0
    \qquad
    \text{for every }e\in U.
\]
The goal is to decompose the instance into nice Bernoulli matroid instances. The extraction procedure repeatedly finds a set of maximum threshold in the
current matroid, extracts the restriction to this set, and contracts it. More
formally, at step \(i\), let \(M_i\) be the current matroid. We choose a
nonempty set \(U_i\) maximizing \(T_{M_i}(S)\), where
\[
    T_{M_i}(S):=\frac{w(S)}{r_{M_i}(S)+x(S)}.
\]
We then define
\[
    M'_i:=M_i|U_i
\]
and continue with
\[
    M_{i+1}:=M_i/U_i.
\]
The procedure stops when the remaining matroid is nice.

\begin{algorithm}[H]
\caption{\textsc{Extract}(\(M,x,v\))}
\label{alg:extract}
\begin{algorithmic}[1]
\State \(i\gets 1\), \(M_1\gets M\)
\While{\(M_i\) is not nice}
    \State Choose nonempty \(U_i\subseteq U(M_i)\) maximizing
    \(T_{M_i}(S):=w(S)/(r_{M_i}(S)+x(S))\)
    \State \(M'_i\gets M_i|U_i\), \(T_i\gets T_{M'_i}(U_i)\),
    \(M_{i+1}\gets M_i/U_i\), \(i\gets i+1\)
\EndWhile
\State \(U_i\gets U(M_i)\), \(M'_i\gets M_i\), \(T_i\gets T_{M'_i}(U_i)\)
\State \Return \(M'_1,\ldots,M'_i\) and \(T_1,\ldots,T_i\)
\end{algorithmic}
\end{algorithm}

We denote the number of output matroids by \(\ell\). Thus the output is
\[
    M'_1,\ldots,M'_\ell,
    \qquad
    U_1,\ldots,U_\ell,
    \qquad
    T_1,\ldots,T_\ell.
\]
The sets \(U_1,\ldots,U_\ell\) form a partition of \(U\). Moreover, for \(i=1,\ldots,\ell\), let
\[
    A_i:=U_1\cup\cdots\cup U_i,
    \qquad \textit{and}\qquad
    A_0:=\emptyset .
\]
then, it can be derived that   $M_i = M/A_{i-1}$, which in turn implies
\[
    M'_i = M_i|U_i = (M/A_{i-1})|U_i .
\]

We next establish the main properties of the extraction procedure that will be
used in the analysis of the algorithm.

\begin{lemma}\label{lem:extraction-polytime}
    Given an independence oracle for \(M\), the extraction procedure can be
implemented in polynomial time. (See the proof in \Cref{app-lem:extraction-polytime}.)
\end{lemma}

\begin{lemma}\label{lem:extraction-nice}
    For every \(i=1,\ldots,\ell\), the Bernoulli instance
\[
    (M'_i,x|_{U_i},v|_{U_i})
\]
is nice.
\end{lemma}

\begin{proof}
Fix \(i\in\{1,\ldots,\ell\}\). If \(i=\ell\) and the extraction procedure stops
because the current matroid is nice, then \(M'_\ell=M_\ell\), so the claim holds
by the stopping condition. Now suppose \(i<\ell\). Then \(U_i\) was chosen to maximize   $ T_{M_i}(S)$ over all nonempty \(S\subseteq U(M_i)\). Since
\[
    M'_i=M_i|U_i,
\]
for every nonempty set \(S\subseteq U_i\) we have
\[
    r_{M'_i}(S)=r_{M_i}(S).
\]
Therefore
\[
    T_{M'_i}(S)=T_{M_i}(S).
\]
By maximality of \(U_i\),
\[
    T_{M_i}(S)\le T_{M_i}(U_i).
\]
Thus
\[
    T_{M'_i}(S)\le T_{M'_i}(U_i)
\]
for every nonempty \(S\subseteq U_i\). Hence
\[
    (M'_i,x|_{U_i},v|_{U_i})
\]
is nice.    
\end{proof}

\begin{lemma}\label{lem:extraction-thresholds}
    The thresholds returned by the extraction procedure satisfy
\[
    T_\ell\le T_{\ell-1}\le\cdots\le T_2\le T_1 .
\]
\end{lemma}

\begin{proof}
We prove that
\[
    T_{i+1}\le T_i
\]
for every \(i=1,\ldots,\ell-1\). This implies
\[
    T_\ell\le T_{\ell-1}\le\cdots\le T_2\le T_1.
\]

Fix \(i<\ell\), and let \(S\subseteq U(M_{i+1})\) be any nonempty set. Since
\[
    M_{i+1}=M_i/U_i,
\]
we have
\[
    r_{M_{i+1}}(S)
    =
    r_{M_i}(U_i\cup S)-r_{M_i}(U_i).
\]
Because \(U_i\) maximizes \(T_{M_i}\), we have
\[
    T_{M_i}(U_i\cup S)\le T_{M_i}(U_i)=T_i.
\]
Therefore,
\begin{equation}
\label{eq:first-threshold-union-bound}
    w(U_i)+w(S)
    \le
    T_i\bigl(r_{M_i}(U_i\cup S)+x(U_i)+x(S)\bigr).
\end{equation}
On the other hand, by definition of \(T_i\),
\begin{equation}
\label{eq:first-threshold-piece-equality}
    w(U_i)=T_i\bigl(r_{M_i}(U_i)+x(U_i)\bigr).
\end{equation}
Subtracting \eqref{eq:first-threshold-piece-equality} from
\eqref{eq:first-threshold-union-bound} gives
\[
    w(S)
    \le
    T_i\bigl(r_{M_i}(U_i\cup S)-r_{M_i}(U_i)+x(S)\bigr).
\]
Using the contraction identity, this becomes
\[
    w(S)
    \le
    T_i\bigl(r_{M_{i+1}}(S)+x(S)\bigr).
\]
Hence
\[
    T_{M_{i+1}}(S)\le T_i
\]
for every nonempty \(S\subseteq U(M_{i+1})\).

In particular, the set chosen at the next step, or the whole remaining ground
set if the algorithm stops at the next step, has threshold at most \(T_i\).
Thus
\[
    T_{i+1}\le T_i.
\]
This proves the lemma.
\end{proof}

\begin{lemma}\label{lem:extraction-rank-dominance}
    For every \(k=1,\ldots,\ell\),
\[
    \sum_{i=1}^k r_{M'_i}(U_i)
    \ge
    x(U_1\cup\cdots\cup U_k).
\]
\end{lemma}
   \begin{proof}
For every \(i\), we have
\[
    r_{M'_i}(U_i)
    =
    r_{M_i}(U_i)
    =
    r_{M/A_{i-1}}(U_i).
\]
By the contraction identity,
\[
    r_{M/A_{i-1}}(U_i)
    =
    r_M(A_i)-r_M(A_{i-1}).
\]
Therefore, for every \(k=1,\ldots,\ell\),
\[
    \sum_{i=1}^k
    r_M(A_i)-r_M(A_{i-1})
    =
    r_M(A_k).
\]
Since \(x\in P_M\), we have
\[
    x(A_k)\le r_M(A_k).
\]
Thus
\[
    \sum_{i=1}^k r_{M'_i}(U_i)
    =
    r_M(A_k)
    \ge
    x(A_k)
    =
    x(U_1\cup\cdots\cup U_k).
\]
This proves the claim.
\end{proof}

\begin{lemma}\label{lem:extraction-charging}
    Let \(T_i\) be the threshold of the component \(M'_i\). Then
\[
    \sum_{i=1}^{\ell} r_{M'_i}(U_i)T_i
    \ge
    \frac12\sum_{e\in U}x_ev_e .
\]
\end{lemma}

\begin{proof}
For each \(i\), let
\[
    X_i:=x(U_i),
    \qquad
    W_i:=w(U_i),
    \qquad
    R_i:=r_{M'_i}(U_i).
\]
By definition of \(T_i\),
\[
    T_i=\frac{W_i}{R_i+X_i}.
\]
Thus
\[
    W_i=(R_i+X_i)T_i.
\]
Therefore,
\[
    \sum_{e\in U}x_ev_e
    =
    \sum_{i=1}^{\ell}W_i
    =
    \sum_{i=1}^{\ell}(R_i+X_i)T_i.
\]
It is enough to prove that
\[
    \sum_{i=1}^{\ell}R_iT_i
    \ge
    \sum_{i=1}^{\ell}X_iT_i,
\]
because then
\[
    \sum_{i=1}^{\ell}R_iT_i
    \ge
    \frac12
    \sum_{i=1}^{\ell}(R_i+X_i)T_i
    =
    \frac12\sum_{e\in U}x_ev_e .
\]

Define
\[
    D_k:=\sum_{i=1}^k (R_i-X_i)
    \qquad
    \text{for }k=1,\ldots,\ell.
\]
By \Cref{lem:extraction-rank-dominance}, we have
\[
    D_k\ge 0
    \qquad
    \text{for every }k=1,\ldots,\ell.
\]
Also, by \Cref{lem:extraction-thresholds},
\[
    T_1\ge T_2\ge\cdots\ge T_\ell\ge 0.
\]
Using summation by parts,
\[
\begin{aligned}
    \sum_{i=1}^{\ell}(R_i-X_i)T_i
    &=
    \sum_{i=1}^{\ell-1}
        D_i(T_i-T_{i+1})
    +
    D_\ell T_\ell .
\end{aligned}
\]
Each term on the right-hand side is nonnegative, since
\[
    D_i\ge 0
    \qquad\text{and}\qquad
    T_i-T_{i+1}\ge 0.
\]
Therefore,
\[
    \sum_{i=1}^{\ell}(R_i-X_i)T_i\ge 0.
\]
Equivalently,
\[
    \sum_{i=1}^{\ell}R_iT_i
    \ge
    \sum_{i=1}^{\ell}X_iT_i.
\]
As explained above, this implies
\[
    \sum_{i=1}^{\ell} r_{M'_i}(U_i)T_i
    =
    \sum_{i=1}^{\ell}R_iT_i
    \ge
    \frac12\sum_{e\in U}x_ev_e .
\]
This completes the proof.
\end{proof}

\subsection{The Algorithm}
\label{subsec:first-online}

We now describe the final online algorithm for a matroid prophet inequality
instance with underlying matroid \(M\).

\begin{algorithm}[H]
\caption{\textsc{FixedThresholdMatroidProphet}}
\label{alg:main}
\begin{algorithmic}[1]
\State Solve the ex-ante relaxation and obtain \(x\in P_M\)
\State Reduce to the Bernoulli instance \((M,x,v)\), and remove all elements \(e\) with \(x_e=0\)
\State Run \textsc{Extract}(\(M,x,v\)) and obtain
\[
    (M'_1,U_1,T_1),\ldots,(M'_\ell,U_\ell,T_\ell).
\]
\State \(I\gets\emptyset\)
\For{each arriving element \(e\) in the Bernoulli instance}
    \State Let \(j\) be the index with \(e\in U_j\)
    \If{\(e\) is active and
    \((I\cap U_j)\cup\{e\}\in\mathcal I(M'_j)\)}\label{line:independence-Mj}
        \State \(I\gets I\cup\{e\}\)
    \EndIf
\EndFor
\State \Return \(I\)
\end{algorithmic}
\end{algorithm}

All thresholds and all additional constraints are computed before the arrival
process begins. During the online phase, the algorithm only checks whether the
arriving element is active and whether accepting it preserves independence in
the extracted matroid corresponding to its component.

\begin{remark}
The condition \(v_e\geq T_j\) does not need to be checked explicitly in the
online algorithm. Indeed, by \Cref{lem:extraction-nice}, each extracted
instance
\[
(M'_j,x|_{U_j},v|_{U_j})
\]
is nice. Hence, by \Cref{lem:nice-values-above-threshold}, every element
\(e\in U_j\) satisfies
\[
v_e\geq T_j.
\]
Thus, on each extracted component, every element automatically lies above the
corresponding threshold. The thresholds are still essential for constructing
the decomposition and for the analysis, but they play no role in the online
acceptance rule once the decomposition has been computed.
\end{remark}

\begin{remark}
The number of extracted matroids is at most \(n\), where \(n=|U|\). Indeed, at
each iteration of \Cref{alg:extract}, at least one element is removed from the
current ground set. Moreover, the independence test in Line~\ref{line:independence-Mj} can be
implemented efficiently. Since each \(M'_j\) is obtained from \(M\) by a sequence
of contractions and restrictions, an independence oracle for \(M\) gives an
independence oracle for \(M'_j\) with only polynomial overhead.
\end{remark}

\begin{theorem}
\label{thm:main}
Let \(I\) be the random set accepted by the algorithm. Then
\[
    \mathbb E[v(I)]
    \ge
    \frac12\sum_{e\in U}x_ev_e .
\]
Consequently, the algorithm is a \(1/2\)-approximation with respect to the
prophet benchmark.
\end{theorem}

\begin{proof}
Let
\[
    M'_1,\ldots,M'_\ell
\]
be the matroids returned by the extraction procedure, with corresponding ground
sets \(U_1,\ldots,U_\ell\) and thresholds \(T_1,\ldots,T_\ell\). By \Cref{lem:direct-sum-submatroid}, every set independent in
\[
M'_1\oplus\cdots\oplus M'_\ell
\]
is independent in $M$. Therefore, any set accepted by the algorithm is
independent in \(M\), so the algorithm is feasible for the original matroid.

By \Cref{lem:extraction-nice}, for each \(i=1,\ldots,\ell\), the Bernoulli
instance
\[
    (M'_i,x|_{U_i},v|_{U_i})
\]
is nice. Moreover, by \Cref{lem:nice-values-above-threshold}, every
\(e\in U_i\) satisfies \(v_e\geq T_i\). Hence, on component \(U_i\), the
algorithm is exactly the single-threshold algorithm with threshold \(T_i\),
with the threshold condition being automatically satisfied. Therefore, by
\Cref{lem:nice-matroid-single-threshold}, if \(I_i:=I\cap U_i\), then
\[
    \mathbb E[v(I_i)]
    \ge
    r_{M'_i}(U_i)T_i .
\]
Since the sets \(U_1,\ldots,U_\ell\) partition \(U\), we have
\[
    v(I)=\sum_{i=1}^{\ell}v(I_i).
\]
Therefore,
\[
\begin{aligned}
    \mathbb E[v(I)]
    &=
    \sum_{i=1}^{\ell}\mathbb E[v(I_i)] \\
    &\ge
    \sum_{i=1}^{\ell} r_{M'_i}(U_i)T_i .
\end{aligned}
\]
By \Cref{lem:extraction-charging},
\[
    \sum_{i=1}^{\ell} r_{M'_i}(U_i)T_i
    \ge
    \frac12\sum_{e\in U}x_ev_e .
\]
Thus,
\[
    \mathbb E[v(I)]
    \ge
    \frac12\sum_{e\in U}x_ev_e .
\]

Finally, \(\sum_{e\in U}x_ev_e\) is the value of the ex-ante relaxation, which
upper bounds the prophet benchmark. Hence the algorithm is a
\(1/2\)-approximation.
\end{proof}

\begin{remark}[Polynomial-time implementation]
By \Cref{lem:extraction-polytime}, the extraction procedure and all component
thresholds can be computed in polynomial time. Together with the polynomial
overhead for simulating independence oracles for the extracted minors discussed
above, this implies that, once the ex-ante solution is available, the entire
procedure of this section runs in polynomial time. Hence whenever the ex-ante
relaxation can be solved efficiently, the full algorithm can be implemented in
polynomial time.
\end{remark}

\section{An Almost Non-Adaptive $\frac{1}{2}$-Approximation via Surplus Thresholds}
\label{sec:surplus-algorithm}

In this section, we present another almost non-adaptive algorithm for matroid
prophet inequalities. The approach closely parallels the one developed in the
previous section. There, we defined a particular threshold for every subset of
a given matroid, used these thresholds to define nice instances, and introduced
an extraction procedure that decomposes an arbitrary instance into nice
instances. We then established several structural properties of this
decomposition and used them to obtain a $1/2$-approximation.

Here, we follow the same framework using a different notion of threshold. We
again associate a threshold with every subset of the matroid, define the
corresponding notion of a nice instance, and introduce an analogous extraction
procedure. We will show that this new extraction procedure satisfies
counterparts of the main properties proved in the previous section and again
yields a $1/2$-approximation for matroid prophet inequalities. The reason for introducing this alternative threshold is that, unlike the
threshold used in the previous section, it admits a natural generalization to
the intersection of several matroids. In the next section, we use this
generalization to obtain a $(q+1)$-approximation for the intersection of $q$
matroids.

\begin{definition}[Surplus threshold of a set]
Let $(M,x,v)$ be a Bernoulli instance, where $M=(U,\mathcal I)$ is a
loopless matroid with rank function $r_M$. For every nonempty set
$S\subseteq U$, define the \emph{surplus threshold} $T'(S)$ to be the unique
nonnegative real number satisfying
\[
T'(S)r_M(S)
=
\sum_{e\in S}x_e\bigl(v_e-T'(S)\bigr)^+.
\]
\end{definition}

We next define the analogue of the nice instances from the previous section
with respect to the new threshold. We call such instances surplus-nice.

\begin{definition}[Surplus-nice Bernoulli matroid instance]
A Bernoulli instance $(M,x,v)$ with $M$ loopless is called
\emph{surplus-nice} if
\[
T'(S)\leq T'(U)
\]
for every nonempty set $S\subseteq U$.
\end{definition}

\begin{remark}
Unlike niceness in the previous section, surplus-niceness does not imply that
$v_e\geq T'(U)$ for every element $e$. For example, consider a rank-one
matroid on two parallel elements $a,b$ with
$x_a=x_b=\frac{1}{2}$ and $(v_a,v_b)=(10,1)$. Then the instance is
surplus-nice and $T'(U)=\frac{10}{3}$, while $v_b<T'(U)$. Thus the threshold
condition $v_e\geq T'$ remains essential in the online algorithm of this
section.
\end{remark}

For a surplus-nice instance, we use the single threshold
\[
T':=T'(U).
\]

We first show that, as in the previous section, surplus-niceness is sufficient
for a single fixed threshold to give the desired local guarantee.

\begin{lemma}[Single-threshold guarantee for surplus-nice matroids]
\label{lem:surplus-nice-single-threshold}
Let $(M,x,v)$ be a surplus-nice Bernoulli matroid instance, and let
\[
T':=T'(U).
\]
Consider the single-threshold algorithm that accepts an active element $e$ if
$v_e\geq T'$ and accepting $e$ preserves independence in $M$. If $I$ denotes
the random independent set selected by the algorithm, then
\[
\mathbb E[v(I)]
\geq
r_M(U)T'.
\]
\end{lemma}

\begin{proof}
Define
\[
y_e:=x_e(v_e-T')^+
\]
for every $e\in U$. By the definition of $T'=T'(U)$,
\[
y(U)
=
\sum_{e\in U}x_e(v_e-T')^+
=
T'r_M(U).
\]

Moreover, since the instance is surplus-nice, for every nonempty set
$S\subseteq U$ we have
\[
T'(S)\leq T'.
\]
The function
\[
t\longmapsto
\sum_{e\in S}x_e(v_e-t)^+-t\,r_M(S)
\]
is non-increasing in $t$ and equals zero at $t=T'(S)$. Therefore,
\[
\sum_{e\in S}x_e(v_e-T')^+
\leq
T'r_M(S).
\]
Equivalently,
\begin{equation}
\label{eq:surplus-density-bound}
y(S)\leq T'r_M(S)
\qquad
\text{for every }S\subseteq U.
\end{equation}

This also holds for $S=\varnothing$. Let $I$ be the random independent set
selected by the algorithm, and let
\[
F:=\operatorname{cl}_M(I)
\]
be its closure in $M$. Let $I'$ be the set of active elements $e$ satisfying
$v_e\geq T'$. Since the algorithm greedily accepts every element of $I'$ whose
addition preserves independence, $I$ is a maximal independent subset of $I'$.
Hence
\[
F=\operatorname{cl}_M(I)=\operatorname{cl}_M(I').
\]
Since $I$ is independent and spans $F$,
\[
r_M(F)=|I|.
\]

Fix an element $e$ with $v_e\geq T'$. If $e$ has no space when it arrives,
then it is spanned by elements that have already been accepted. Since these
elements are contained in $I$, it follows that $e\in F$. Hence
\[
\Pr[e\text{ has space}]
\geq
1-\Pr[e\in F].
\]
The arrival order is fixed before the Bernoulli realizations. Therefore, the
event that $e$ has space depends only on the activations of elements preceding
$e$ in this fixed order and is independent of the activation of $e$. Thus,
\[
\begin{aligned}
\mathbb E\left[\sum_{e\in I}(v_e-T')\right]
&=
\sum_{e\in U}x_e(v_e-T')^+\Pr[e\text{ has space}]\\
&\geq
\sum_{e\in U}x_e(v_e-T')^+\bigl(1-\Pr[e\in F]\bigr)\\
&=
y(U)-\mathbb E[y(F)].
\end{aligned}
\]

Consequently,
\[
\begin{aligned}
\mathbb E[v(I)]
&=
T'\mathbb E[|I|]
+
\mathbb E\left[\sum_{e\in I}(v_e-T')\right]\\
&\geq
y(U)+\mathbb E\left[T'r_M(F)-y(F)\right]\\
&\geq
y(U).
\end{aligned}
\]
The last inequality follows from \eqref{eq:surplus-density-bound} applied to
the random set $F$. Finally, $y(U)=T'r_M(U)$, and hence
\[
\mathbb E[v(I)]
\geq
T'r_M(U).
\]
\end{proof}

\subsection{The Extraction Procedure}
\label{subsec:surplus-extraction}

We now describe the extraction procedure. Given a Bernoulli instance
$(M,x,v)$ with $M$ loopless, the procedure repeatedly extracts a set of
maximum surplus threshold, restricts the current matroid to this set, and
contracts it. The following lemma shows that a maximum-threshold set can always
be chosen to be a flat.

\begin{lemma}
\label{lem:surplus-max-flat}
Let $(M,x,v)$ be a Bernoulli instance with $M$ loopless. Then the maximum of
$T'(S)$ over all nonempty sets $S\subseteq U$ is attained by a flat of $M$.
More precisely, for every nonempty set $S\subseteq U$,
\[
T'(S)
\leq
T'(\operatorname{cl}_M(S)).
\]
\end{lemma}

\begin{proof}
Fix a nonempty set $S\subseteq U$, and let
\[
F:=\operatorname{cl}_M(S).
\]
Since taking the closure does not change the rank,
\[
r_M(F)=r_M(S).
\]

Let
\[
T:=T'(S).
\]
By the definition of the surplus threshold,
\[
Tr_M(S)
=
\sum_{e\in S}x_e(v_e-T)^+.
\]
Since $S\subseteq F$ and all terms in the sum are nonnegative,
\begin{equation}
\label{eq:surplus-closure-bound}
\sum_{e\in F}x_e(v_e-T)^+
\geq
\sum_{e\in S}x_e(v_e-T)^+
=
Tr_M(S)
=
Tr_M(F).
\end{equation}

Now define
\[
g_F(t)
:=
\sum_{e\in F}x_e(v_e-t)^+
-
t r_M(F).
\]
By definition, $T'(F)$ is the unique nonnegative value satisfying
\[
g_F(T'(F))=0.
\]
By \eqref{eq:surplus-closure-bound},
\[
g_F(T)\geq 0.
\]
Since $F$ is nonempty and $M$ is loopless, we have $r_M(F)>0$. Hence
$g_F$ is strictly decreasing. It follows that
\[
T'(F)\geq T=T'(S).
\]

Therefore,
\[
T'(S)
\leq
T'(\operatorname{cl}_M(S)).
\]
Consequently, the closure of any maximum-threshold set is also a
maximum-threshold set, and hence a maximum surplus threshold is attained by a
flat.
\end{proof}

\begin{remark}
The flat property above is specific to the surplus threshold and need not hold
for the threshold used in \Cref{sec:first-algorithm}. For example, in the rank-one matroid on
two parallel elements $a,b$ with $x_a=x_b=\frac{1}{2}$ and
$(v_a,v_b)=(10,1)$, the threshold of \Cref{sec:first-algorithm} satisfies
$T(\{a\})=\frac{10}{3}$, whereas
$T(\operatorname{cl}_M(\{a\}))=T(\{a,b\})=\frac{11}{4}$. Thus taking the
closure can decrease the old threshold, and a maximizing set need not be a
flat.
\end{remark}

By \Cref{lem:surplus-max-flat}, at every iteration we may choose a
maximum-threshold set that is a flat of the current matroid. Since contracting
a flat does not create loops, every residual matroid produced by the extraction
procedure remains loopless.

\begin{algorithm}[H]
\caption{\textsc{Surplus-Extract}$(M,x,v)$}
\label{alg:Surplus-Extract}
\begin{algorithmic}[1]
\State $i\gets 1$, $M_i\gets M$
\While{$M_i$ is not surplus-nice}
\State Choose a nonempty flat $U_i\subseteq U(M_i)$ maximizing
$T'_{M_i}(S)$ over all nonempty sets $S\subseteq U(M_i)$
\State $M'_i\gets M_i|U_i$, $T'_i\gets T'_{M'_i}(U_i)$,
$M_{i+1}\gets M_i/U_i$, $i\gets i+1$
\EndWhile
\State $U_i\gets U(M_i)$, $M'_i\gets M_i$,
$T'_i\gets T'_{M'_i}(U_i)$
\State \Return $(M'_1,U_1,T'_1),\ldots,(M'_\ell,U_\ell,T'_\ell)$
\end{algorithmic}
\end{algorithm}

Let the output of the procedure be
\[
(M'_1,U_1,T'_1),\ldots,(M'_\ell,U_\ell,T'_\ell).
\]
The sets
\[
U_1,\ldots,U_\ell
\]
form a partition of $U$. Moreover, define
\[
A_i:=U_1\cup\cdots\cup U_i,
\qquad
A_0:=\emptyset.
\]
By construction,
\[
M_i=M/A_{i-1},
\]
and therefore
\[
M'_i
=
M_i|U_i
=
(M/A_{i-1})|U_i.
\]

Since each $U_i$ is a flat of the current matroid $M_i$, the contraction
\[
M_{i+1}=M_i/U_i
\]
is loopless whenever $M_i$ is loopless. Thus, by induction, every residual
matroid $M_i$ appearing in the extraction procedure is loopless.

We next establish the main properties of the extraction procedure.

\begin{lemma}
\label{lem:extraction-Surplus-polytime}
Given an independence oracle for $M$, the extraction procedure can be
implemented in polynomial time. Moreover, at every iteration, a
maximum-surplus-threshold flat can be found in polynomial time. (See the proof in \Cref{app-lem:extraction-Surplus-polytime}.)
\end{lemma}

\begin{lemma}
For every $i=1,\ldots,\ell$, the Bernoulli instance
\[
(M'_i,x|_{U_i},v|_{U_i})
\]
is surplus-nice.
\end{lemma}

\begin{proof}
Fix $i\in\{1,\ldots,\ell\}$. By construction, every residual matroid $M_i$
is loopless. Since
\[
M'_i=M_i|U_i,
\]
and restriction does not create loops, $M'_i$ is also loopless. If $i=\ell$, then the extraction procedure terminates because the remaining
instance is surplus-nice, and the claim follows immediately. Now suppose $i<\ell$. By construction, $U_i$ maximizes
\[
T'_{M_i}(S)
\]
over all nonempty sets $S\subseteq U(M_i)$. Since
\[
M'_i=M_i|U_i,
\]
for every nonempty set $S\subseteq U_i$ we have
\[
r_{M'_i}(S)=r_{M_i}(S),
\]
and therefore
\[
T'_{M'_i}(S)=T'_{M_i}(S).
\]

By the maximality of $U_i$,
\[
T'_{M_i}(S)\leq T'_{M_i}(U_i).
\]
Hence,
\[
T'_{M'_i}(S)
\leq
T'_{M'_i}(U_i)
\]
for every nonempty set $S\subseteq U_i$. Thus $M'_i$ is loopless and
\[
T'_{M'_i}(S)\leq T'_{M'_i}(U_i)
\]
for every nonempty $S\subseteq U_i$. Therefore,
\[
(M'_i,x|_{U_i},v|_{U_i})
\]
is surplus-nice.
\end{proof}

\begin{lemma}
The extracted thresholds satisfy
\[
T'_1\geq T'_2\geq\cdots\geq T'_\ell.
\]
\end{lemma}

\begin{proof}
We show that
\[
T'_{i+1}\leq T'_i
\]
for every $i=1,\ldots,\ell-1$. Fix $i$, and consider any nonempty set
\[
S\subseteq U(M_{i+1}).
\]
Since every residual matroid is loopless, $r_{M_{i+1}}(S)>0$ and
$T'_{M_{i+1}}(S)$ is well defined. By construction,
\[
M_{i+1}=M_i/U_i,
\]
and hence
\[
r_{M_{i+1}}(S)
=
r_{M_i}(U_i\cup S)-r_{M_i}(U_i).
\]

Since $U_i$ maximizes the surplus threshold in $M_i$,
\[
T'_{M_i}(U_i\cup S)\leq T'_i.
\]
Therefore, by the definition of the surplus threshold,
\begin{equation}
\label{eq:surplus-union-bound}
\sum_{e\in U_i\cup S}x_e(v_e-T'_i)^+
\leq
T'_i r_{M_i}(U_i\cup S).
\end{equation}
On the other hand, since
\[
T'_i=T'_{M_i}(U_i),
\]
we have
\begin{equation}
\label{eq:surplus-piece-equality}
\sum_{e\in U_i}x_e(v_e-T'_i)^+
=
T'_i r_{M_i}(U_i).
\end{equation}
Subtracting \eqref{eq:surplus-piece-equality} from
\eqref{eq:surplus-union-bound} gives
\[
\sum_{e\in S}x_e(v_e-T'_i)^+
\leq
T'_i
\bigl(
r_{M_i}(U_i\cup S)-r_{M_i}(U_i)
\bigr)
=
T'_i r_{M_{i+1}}(S).
\]
It follows that
\begin{equation}
\label{eq:surplus-residual-threshold-bound}
T'_{M_{i+1}}(S)\leq T'_i.
\end{equation}

If $i+1<\ell$, then $U_{i+1}$ maximizes the surplus threshold in
$M_{i+1}$, and therefore
\[
T'_{i+1}
=
T'_{M_{i+1}}(U_{i+1})
\leq T'_i.
\]
If $i+1=\ell$, then $U_\ell=U(M_\ell)$, so applying \eqref{eq:surplus-residual-threshold-bound}
with $S=U_\ell$ gives
\[
T'_\ell\leq T'_i.
\]

Thus,
\[
T'_1\geq T'_2\geq\cdots\geq T'_\ell.
\]
\end{proof}

\begin{lemma}
For every \(k=1,\ldots,\ell\),
\[
    \sum_{i=1}^k r_{M'_i}(U_i)
    \ge
    x(U_1\cup\cdots\cup U_k).
\]
\end{lemma}

\begin{proof}
This is exactly the same argument as in the corresponding \cref{lem:extraction-rank-dominance}, since it depends only on the contraction structure of the
extraction and on $x\in P_M$, and not on the definition of the threshold.
\end{proof}

\begin{lemma}
The extracted thresholds satisfy
\[
\sum_{i=1}^{\ell} r_{M'_i}(U_i)T'_i
\geq
\frac{1}{2}\sum_{e\in U}x_ev_e.
\]
\end{lemma}

\begin{proof}
For every $i$, let
\[
R_i:=r_{M'_i}(U_i),
\qquad
X_i:=x(U_i).
\]
Since
\[
T'_i=T'_{M'_i}(U_i),
\]
the definition of the surplus threshold gives
\[
R_iT'_i
=
\sum_{e\in U_i}x_e(v_e-T'_i)^+.
\]

For every $e\in U_i$,
\[
v_e
=
(v_e-T'_i)^+
+
\min\{v_e,T'_i\}
\leq
(v_e-T'_i)^+ + T'_i.
\]
Multiplying by $x_e$ and summing over $U_i$, we obtain
\[
\sum_{e\in U_i}x_ev_e
\leq
R_iT'_i+X_iT'_i.
\]
Therefore,
\[
\sum_{e\in U}x_ev_e
\leq
\sum_{i=1}^{\ell}R_iT'_i
+
\sum_{i=1}^{\ell}X_iT'_i.
\]

It remains to compare the two terms on the right. By the previous lemma,
for every $k=1,\ldots,\ell$,
\[
\sum_{i=1}^k X_i
\leq
\sum_{i=1}^k R_i,
\]
and by the monotonicity lemma,
\[
T'_1\geq T'_2\geq\cdots\geq T'_\ell\geq 0.
\]
Hence, by summation by parts,
\[
\sum_{i=1}^{\ell}X_iT'_i
\leq
\sum_{i=1}^{\ell}R_iT'_i.
\]
Indeed, writing
\[
D_k:=\sum_{i=1}^k(R_i-X_i)\geq 0,
\]
we have
\[
\sum_{i=1}^{\ell}(R_i-X_i)T'_i
=
\sum_{k=1}^{\ell-1}
D_k(T'_k-T'_{k+1})
+
D_\ell T'_\ell
\geq 0.
\]

Consequently,
\[
\sum_{e\in U}x_ev_e
\leq
2\sum_{i=1}^{\ell}R_iT'_i.
\]
Substituting $R_i=r_{M'_i}(U_i)$ gives
\[
\sum_{i=1}^{\ell}r_{M'_i}(U_i)T'_i
\geq
\frac{1}{2}\sum_{e\in U}x_ev_e,
\]
as claimed.
\end{proof}

\subsection{The Algorithm}
\label{subsec:surplus-online}

We now describe the final online algorithm. All thresholds and all additional
constraints are computed before the arrival process begins.

\begin{algorithm}[H]
\caption{\textsc{FixedThresholdMatroidProphet-Surplus}}
\begin{algorithmic}[1]
\State Solve the ex-ante relaxation and obtain $x\in P_M$.
\State Reduce to the Bernoulli instance $(M,x,v)$ and remove all elements
with $x_e=0$.
\State Run \textsc{Surplus-Extract} and obtain
\[
(M'_1,U_1,T'_1),\ldots,(M'_\ell,U_\ell,T'_\ell).
\]
\State $I\gets\emptyset$.
\For{each arriving element $e$}
\State Let $j$ be the index such that $e\in U_j$.
\If{$e$ is active, $v_e\geq T'_j$, and
\[
(I\cap U_j)\cup\{e\}\in\mathcal I(M'_j)
\]
}
\State $I\gets I\cup\{e\}$.
\EndIf
\EndFor
\State \Return $I$.
\end{algorithmic}
\end{algorithm}

\begin{theorem}
The algorithm above satisfies
\[
\mathbb E[v(I)]
\geq
\frac{1}{2}\sum_{e\in U}x_ev_e.
\]
Consequently, it is a $\frac{1}{2}$-approximation with respect to the prophet
benchmark.
\end{theorem}

\begin{proof}
Let
\[
(M'_1,U_1,T'_1),\ldots,(M'_\ell,U_\ell,T'_\ell)
\]
be the output of \textsc{Surplus-Extract}. By the direct-sum submatroid \cref{lem:direct-sum-submatroid}, every set that is
independent in
\[
M'_1\oplus\cdots\oplus M'_\ell
\]
is independent in the original matroid $M$. Hence the set accepted by the
algorithm is feasible in $M$. For every $i=1,\ldots,\ell$, the Bernoulli instance
\[
(M'_i,x|_{U_i},v|_{U_i})
\]
is surplus-nice. Therefore, by the single-threshold guarantee for
surplus-nice instances,
\[
\mathbb E[v(I\cap U_i)]
\geq
r_{M'_i}(U_i)T'_i.
\]

Since $U_1,\ldots,U_\ell$ form a partition of $U$,
\[
\begin{aligned}
\mathbb E[v(I)]
&=
\sum_{i=1}^{\ell}\mathbb E[v(I\cap U_i)]\\
&\geq
\sum_{i=1}^{\ell}r_{M'_i}(U_i)T'_i.
\end{aligned}
\]

By the preceding lemma,
\[
\sum_{i=1}^{\ell}r_{M'_i}(U_i)T'_i
\geq
\frac{1}{2}\sum_{e\in U}x_ev_e.
\]
Therefore,
\[
\mathbb E[v(I)]
\geq
\frac{1}{2}\sum_{e\in U}x_ev_e.
\]

Finally, the ex-ante relaxation upper-bounds the prophet benchmark, so the
algorithm is a $\frac{1}{2}$-approximation with respect to the prophet.
\end{proof}

\subsection{The Surplus Vector}
\label{subsec:surplus-vector}

We conclude this section by recording a reformulation of the preceding
decomposition that will be useful for matroid intersection. For every
$e\in U_i$, define
\[
T'_e:=T'_i
\]
and
\[
y_e:=x_e(v_e-T'_e)^+
=
x_e(v_e-T'_i)^+.
\]
We refer to $y$ as the \emph{surplus vector} associated with the extracted
decomposition.

For every extracted component $U_i$, the definition of $T'_i$ gives
\[
y(U_i)
=
\sum_{e\in U_i}x_e(v_e-T'_i)^+
=
T'_i r_{M'_i}(U_i).
\]
Moreover, since $(M'_i,x|_{U_i},v|_{U_i})$ is surplus-nice, for every
$S\subseteq U_i$,
\[
y(S)
=
\sum_{e\in S}x_e(v_e-T'_i)^+
\leq
T'_i r_{M'_i}(S).
\]
In particular,
\[
y(U)
=
\sum_{e\in U}x_e(v_e-T'_e)^+.
\]

This viewpoint is the bridge to the matroid-intersection setting. For a
single matroid, the extraction procedure determines the thresholds
$T'_e$, and these thresholds in turn determine the surplus vector $y$. In
the next section, we will instead seek a common surplus vector $y$ that
simultaneously induces suitable decompositions for all $q$ matroids. Each
matroid will contribute a price to every element, and the sum of these
prices will determine its threshold. This leads to a coupled fixed-point
condition between the surplus vector and the resulting thresholds.

\begin{remark}[Polynomial-time implementation]
By \Cref{lem:extraction-Surplus-polytime}, a maximum-surplus-threshold flat can
be found and the entire \textsc{Surplus-Extract} decomposition can be computed
in polynomial time. Restrictions, contractions, threshold evaluations, and the
online independence tests can all be implemented with polynomial overhead from
an independence oracle for the original matroid. Consequently, once the
ex-ante solution is available, the whole procedure of this section runs in
polynomial time.
\end{remark}

\section{A $(q+1)$-Approximation for Bernoulli Prophet Inequalities under
Matroid Intersection}
\label{sec:q-intersection}

In this section, we generalize the construction of the previous section to
intersections of matroids. Let
\[
M_i=(U,\mathcal I_i),
\qquad i=1,\ldots,q,
\]
be matroids on a common finite ground set $U$, and let
\[
x\in\bigcap_{i=1}^q P(M_i).
\]
For every element $e\in U$, independently,
\[
X_e=
\begin{cases}
v_e,&\text{with probability }x_e,\\
0,&\text{with probability }1-x_e.
\end{cases}
\]

As before, we may remove all elements with $x_e=0$. In particular, after
this removal every matroid $M_i$ is loopless, since $x\in P(M_i)$ implies
$x_e=0$ for every loop $e$ of $M_i$.

We write
\[
W:=\sum_{e\in U}x_ev_e.
\]

We construct an almost non-adaptive algorithm for the intersection of the
$q$ matroids. The main new ingredient is a single surplus vector $y$ that
simultaneously induces a density decomposition in each matroid. Each
decomposition assigns a price to every element, and the threshold of an
element will be the sum of its prices across the $q$ matroids.

\subsection{A Coupled Density Decomposition and the Surplus Vector}
\label{subsec:coupled-density}

We first define a canonical way of assigning prices to the elements of a
matroid from a nonnegative vector. Let $M=(U,\mathcal I)$ be a loopless matroid with rank function $r_M$, and
let $y\in\mathbb R_+^U$. For every nonempty set $S\subseteq U$, define its
density by
\[
\rho_M(S;y):=\frac{y(S)}{r_M(S)}.
\]

We repeatedly extract a maximum-density set and contract it. At every step,
among all nonempty subsets of the current matroid, we choose the
inclusion-wise maximal maximizer of the density. As we show below, this set
is a flat. Consequently, contracting it does not create loops, and every
residual matroid remains loopless. If the selected set is $V$, we assign the price
\[
p(V):=\frac{y(V)}{r_M(V)}
\]
to every element of $V$, contract $V$, and continue.

\begin{lemma}[Canonical maximum-density set]
\label{lem:density-union}
\label{lem:density-max-flat}
Let
\[
\lambda
:=
\max_{\emptyset\neq S\subseteq U}
\frac{y(S)}{r_M(S)}.
\]
If $A$ and $B$ both have density $\lambda$, then $A\cup B$ also has density
$\lambda$. Consequently, there is a unique inclusion-wise maximal
maximum-density set, and this set is a flat of $M$.
\end{lemma}

\begin{proof}
Since $\lambda$ is the maximum density,
\[
y(A\cap B)\leq \lambda r_M(A\cap B).
\]
Therefore,
\begin{equation}
\label{eq:density-union-lower-bound}
\begin{aligned}
y(A\cup B)
&=y(A)+y(B)-y(A\cap B)\\
&\geq
\lambda r_M(A)+\lambda r_M(B)
-\lambda r_M(A\cap B)\\
&\geq
\lambda r_M(A\cup B).
\end{aligned}
\end{equation}
The final inequality in~\eqref{eq:density-union-lower-bound} follows from
the submodularity of the rank function. By maximality of $\lambda$, the
reverse inequality to~\eqref{eq:density-union-lower-bound} also holds. Hence
\[
y(A\cup B)=\lambda r_M(A\cup B),
\]
so $A\cup B$ is also a maximum-density set. Since the ground set is finite,
the union of all maximum-density sets is therefore the unique inclusion-wise
maximal maximum-density set; denote it by $V$.

It remains to show that $V$ is a flat. Let
\[
F:=\operatorname{cl}_M(V).
\]
Since taking the closure does not change the rank,
\[
r_M(F)=r_M(V).
\]
Moreover, since $V\subseteq F$ and $y\geq 0$,
\[
y(F)\geq y(V).
\]
Thus
\[
\frac{y(F)}{r_M(F)}
\geq
\frac{y(V)}{r_M(V)}
=
\lambda.
\]
By maximality of $\lambda$, equality must hold, so $F$ is also a
maximum-density set. The inclusion-wise maximality of $V$ then gives $F=V$.
Hence $V$ is a flat.
\end{proof}

It follows that the density decomposition produces an ordered partition
\[
U=V_1\dot\cup\cdots\dot\cup V_m
\]
with corresponding prices
\[
p_1\geq p_2\geq\cdots\geq p_m\geq 0.
\]
For a vector $y$, denote the resulting element-wise price vector by
\[
D_M(y)\in\mathbb R_+^U.
\]

\begin{lemma}[Continuity of the density-price map]
\label{lem:density-price-continuity}
For every loopless matroid $M$, the map
\[
D_M:\mathbb R_+^U\to\mathbb R_+^U
\]
is continuous.
\end{lemma}

\begin{proof}
Let
\[
y^{(n)}\to y.
\]
We show that
\[
D_M(y^{(n)})\to D_M(y).
\]

Since the ground set is finite, there are only finitely many ordered
partitions of $U$. Thus, after passing to a subsequence, we may assume that
the density decomposition of every $y^{(n)}$ has the same ordered blocks
\[
V_1,\ldots,V_m.
\]

For each $j$, let
\[
A_{j-1}:=V_1\cup\cdots\cup V_{j-1},
\qquad
N_j:=(M/A_{j-1})|V_j.
\]
The price assigned to the elements of $V_j$ under $y^{(n)}$ is
\[
p_j^{(n)}
=
\frac{y^{(n)}(V_j)}{r_{N_j}(V_j)}.
\]
Since the denominator is fixed and positive,
\[
p_j^{(n)}
\to
p_j:=
\frac{y(V_j)}{r_{N_j}(V_j)}.
\]

Moreover, because $V_j$ is a maximum-density set in the corresponding
residual matroid under $y^{(n)}$, for every nonempty
$S\subseteq U(N_j)$,
\[
y^{(n)}(S)
\leq
p_j^{(n)}r_{N_j}(S).
\]
Taking limits gives
\[
y(S)\leq p_jr_{N_j}(S).
\]
For $S=V_j$, equality holds:
\[
y(V_j)=p_jr_{N_j}(V_j).
\]
Thus $V_j$ is a maximum-density set for $y$ in the residual matroid $N_j$. The decomposition
\[
V_1,\ldots,V_m
\]
need not itself be the canonical decomposition of $y$, because consecutive
blocks whose limiting prices coincide may merge. However, these are the only
possible changes. Indeed, whenever
\[
p_j>p_{j+1},
\]
the strict inequality separates the two blocks in the canonical
maximum-density decomposition of $y$. If instead
\[
p_j=p_{j+1},
\]
the canonical decomposition may merge $V_j$ and $V_{j+1}$, but every element
in the merged block still receives the same price. Consequently, every element $e\in V_j$ receives price $p_j$ under
$D_M(y)$. Since under $D_M(y^{(n)})$ it receives price $p_j^{(n)}$, we obtain
\[
D_M(y^{(n)})_e\to D_M(y)_e
\]
for every $e\in U$. Hence,
\[
D_M(y^{(n)})\to D_M(y),
\]
and $D_M$ is continuous.
\end{proof}

For a candidate surplus vector $y\in\mathbb R_+^U$, define
\[
p^i(y):=D_{M_i}(y),
\qquad i=1,\ldots,q.
\]
For every element $e\in U$, define its total threshold by
\[
T_e(y):=\sum_{i=1}^q p^i_e(y).
\]

We now show that there exists a surplus vector that is consistent with the
thresholds induced by its own density decompositions.

\begin{theorem}[Coupled density fixed point]
\label{thm:coupled-density-fixed-point}
There exists a vector
\[
y^*\in\mathbb R_+^U
\]
such that, if
\[
p^i:=D_{M_i}(y^*)
\]
for every $i=1,\ldots,q$, and
\[
T_e:=\sum_{i=1}^q p^i_e,
\]
then
\[
y^*_e=x_e(v_e-T_e)^+
\]
for every $e\in U$.
\end{theorem}

\begin{proof}
Consider the compact convex set
\[
K:=\prod_{e\in U}[0,x_ev_e].
\]
Define a map
\[
F:K\to K
\]
coordinate-wise by
\[
F_e(y)
:=
x_e
\left(
v_e-\sum_{i=1}^q D_{M_i}(y)_e
\right)^+.
\]

For every $e\in U$ and every $y\in K$,
\[
0\leq F_e(y)\leq x_ev_e,
\]
and hence
\[
F(K)\subseteq K.
\]

By \Cref{lem:density-price-continuity}, each map
\[
D_{M_i}:\mathbb R_+^U\to\mathbb R_+^U
\]
is continuous. Therefore, $F$ is continuous. By Brouwer's fixed-point theorem, there exists
\[
y^*\in K
\]
such that
\[
F(y^*)=y^*.
\]
Thus, for every $e\in U$,
\[
y^*_e
=
x_e
\left(
v_e-\sum_{i=1}^q D_{M_i}(y^*)_e
\right)^+.
\]

Setting
\[
p^i:=D_{M_i}(y^*)
\]
and
\[
T_e:=\sum_{i=1}^q p^i_e
\]
gives
\[
y^*_e=x_e(v_e-T_e)^+,
\]
as desired.
\end{proof}

The theorem above establishes the existence of a surplus vector satisfying the
coupled fixed-point condition. In \Cref{lem:coupled-fixed-point-polytime}, we
show that such a vector can be computed in polynomial time. Fix such a vector
\[
y:=y^*
\]
for the remainder of the section. For each matroid $M_i$, let
\[
V^i_1,\ldots,V^i_{m_i}
\]
be the density decomposition induced by $y$, with corresponding prices
\[
p^i_1\geq\cdots\geq p^i_{m_i}\geq 0.
\]
For every $e\in V^i_j$, define
\[
p^i_e:=p^i_j.
\]
Finally, define
\[
T_e:=\sum_{i=1}^q p^i_e.
\]
By the fixed-point identity,
\[
y_e=x_e(v_e-T_e)^+
\]
for every $e\in U$.

\subsection{The Extracted Matroids}
\label{subsec:extracted-matroids}

Fix the surplus vector $y$ obtained from
\Cref{thm:coupled-density-fixed-point}. For each matroid $M_i$, let
\[
V^i_1,\ldots,V^i_{m_i}
\]
be its density decomposition. For every $j=1,\ldots,m_i$, define
\[
A^i_{j-1}
:=
V^i_1\cup\cdots\cup V^i_{j-1},
\qquad
A^i_0:=\emptyset,
\]
and let
\[
N^i_j
:=
\left(M_i/A^i_{j-1}\right)|V^i_j.
\]

We define the stricter matroid
\[
\widehat M_i
:=
\bigoplus_{j=1}^{m_i} N^i_j.
\]

Since every block $V^i_j$ is a flat of the corresponding residual matroid,
each contraction preserves looplessness. In particular, every matroid
$N^i_j$ is loopless.

\begin{lemma}[Feasibility of the extracted matroids]
\label{lem:intersection-extracted-feasibility}
For every $i=1,\ldots,q$, every set independent in $\widehat M_i$ is
independent in $M_i$.
\end{lemma}

\begin{proof}
This follows directly from the direct-sum submatroid \cref{lem:direct-sum-submatroid}, applied to the decomposition
\[
V^i_1,\ldots,V^i_{m_i}.
\]
Indeed,
\[
N^i_j
=
\left(M_i/A^i_{j-1}\right)|V^i_j,
\]
so every set independent in
\[
\widehat M_i
=
\bigoplus_{j=1}^{m_i}N^i_j
\]
is independent in $M_i$.
\end{proof}

\begin{lemma}[Density inequalities]
\label{lem:intersection-density-inequalities}
For every $i=1,\ldots,q$ and every $j=1,\ldots,m_i$,
\[
y(V^i_j)
=
p^i_j r_{N^i_j}(V^i_j).
\]
Moreover, for every subset
\[
S\subseteq V^i_j,
\]
we have
\[
y(S)
\leq
p^i_j r_{N^i_j}(S).
\]
\end{lemma}

\begin{proof}
By definition, $V^i_j$ is a maximum-density set in the residual matroid
\[
M_i/A^i_{j-1},
\]
and
\[
p^i_j
=
\frac{y(V^i_j)}{r_{N^i_j}(V^i_j)}.
\]
This gives
\[
y(V^i_j)
=
p^i_j r_{N^i_j}(V^i_j).
\]

For every nonempty $S\subseteq V^i_j$, maximality of the density of
$V^i_j$ gives
\[
\frac{y(S)}{r_{N^i_j}(S)}
\leq
p^i_j.
\]
Therefore,
\[
y(S)
\leq
p^i_j r_{N^i_j}(S).
\]
The inequality is trivial for $S=\emptyset$.
\end{proof}

\begin{lemma}[Polynomial-time density decomposition]
\label{lem:matroid-intersection-poly-time}
Given an independence oracle for a loopless matroid $M$ and a vector
$y\in\mathbb R_+^U$, its density decomposition
\[
V_1,\ldots,V_m
\]
and the corresponding prices
\[
p_1,\ldots,p_m
\]
can be computed in polynomial time. (See the proof in \Cref{app-lem:matroid-intersection-poly-time}.)
\end{lemma}

\subsection{The Online Algorithm}
\label{subsec:intersection-online}

The algorithm uses the fixed thresholds
\[
T_e=\sum_{i=1}^q p^i_e.
\]
All thresholds and all additional matroid constraints are determined before
the arrival process begins.

\begin{algorithm}[H]
\caption{\textsc{FixedThresholdMatroidIntersectionProphet}}
\begin{algorithmic}[1]
\State Compute a surplus vector $y$ satisfying the coupled fixed-point
condition, using \Cref{lem:coupled-fixed-point-polytime}.
\State For every $i=1,\ldots,q$, compute the density decomposition of $M_i$
and construct $\widehat M_i$.
\State $I\gets\emptyset$.
\For{each arriving element $e$}
\If{$e$ is active, $v_e\geq T_e$, and $I\cup\{e\}$ is independent in every
$\widehat M_i$}
\State $I\gets I\cup\{e\}$.
\EndIf
\EndFor
\State \Return $I$.
\end{algorithmic}
\end{algorithm}

We now analyze the algorithm. For every $i=1,\ldots,q$ and every density
block $V^i_j$, let
\[
I^i_j:=I\cap V^i_j
\]
and define
\[
F^i_j:=\operatorname{cl}_{N^i_j}(I^i_j).
\]
Since $I^i_j$ is independent in $N^i_j$,
\[
r_{N^i_j}(F^i_j)=|I^i_j|.
\]
Therefore, by the density inequality,
\begin{equation}
\label{eq:intersection-closure-density}
y(F^i_j)
\leq
p^i_j r_{N^i_j}(F^i_j)
=
p^i_j|I^i_j|.
\end{equation}

\begin{lemma}[Value versus surplus]
\label{lem:intersection-gain-surplus}
The online algorithm satisfies
\[
\mathbb E[v(I)]
\geq
y(U).
\]
\end{lemma}

\begin{proof}
Fix an element $e$ with $v_e\geq T_e$. For every matroid $M_i$, let $j_i(e)$
be the index such that
\[
e\in V^i_{j_i(e)}.
\]

If $e$ does not have space when it arrives, then adding $e$ violates
independence in at least one of the matroids $\widehat M_i$. For such a
matroid $i$, the element $e$ lies in the closure of the elements already
accepted in its block $V^i_{j_i(e)}$. Since the accepted set only grows,
this implies
\[
e\in F^i_{j_i(e)}.
\]
Consequently,
\[
\Pr[e\text{ has space}]
\geq
1-
\sum_{i=1}^q
\Pr\left[e\in F^i_{j_i(e)}\right].
\]

Recall that
\[
y_e=x_e(v_e-T_e)^+.
\]
The arrival order is fixed before the Bernoulli realizations. Therefore,
whether $e$ has space when it arrives depends only on the activations of
elements preceding $e$ in this fixed order and is independent of the
activation of $e$. Hence the expected surplus contribution is at least
\[
\begin{aligned}
\sum_{e\in U}
y_e\Pr[e\text{ has space}]
&\geq
y(U)
-
\sum_{e\in U}
y_e
\sum_{i=1}^q
\Pr\left[e\in F^i_{j_i(e)}\right]\\
&=
y(U)
-
\sum_{i=1}^q
\sum_j
\mathbb E\left[y(F^i_j)\right].
\end{aligned}
\]

On the other hand, every accepted element contributes its threshold $T_e$.
Thus the expected threshold contribution is
\[
\begin{aligned}
\mathbb E\left[\sum_{e\in I}T_e\right]
&=
\mathbb E\left[
\sum_{e\in I}\sum_{i=1}^q p^i_e
\right]\\
&=
\sum_{i=1}^q
\sum_j
p^i_j\,\mathbb E[|I^i_j|].
\end{aligned}
\]

For every realization,~\eqref{eq:intersection-closure-density} gives
\[
y(F^i_j)\leq p^i_j|I^i_j|.
\]
Therefore,
\[
\sum_{i=1}^q\sum_j
\mathbb E[y(F^i_j)]
\leq
\sum_{i=1}^q\sum_j
p^i_j\,\mathbb E[|I^i_j|].
\]

Combining the threshold and surplus contributions gives
\[
\begin{aligned}
\mathbb E[v(I)]
&\geq
\sum_{i=1}^q\sum_j
p^i_j\,\mathbb E[|I^i_j|]
+
y(U)
-
\sum_{i=1}^q\sum_j
\mathbb E[y(F^i_j)]\\
&\geq
y(U).
\end{aligned}
\]
\end{proof}

\begin{lemma}[$(q+1)$ accounting]
\label{lem:q-plus-one-accounting}
The surplus vector satisfies
\[
\sum_{e\in U}x_ev_e
\leq
(q+1)y(U).
\]
\end{lemma}

\begin{proof}
For every element $e\in U$,
\[
v_e
\leq
T_e+(v_e-T_e)^+.
\]
Multiplying by $x_e$ and summing over all elements gives
\[
\sum_{e\in U}x_ev_e
\leq
\sum_{e\in U}x_eT_e
+
\sum_{e\in U}x_e(v_e-T_e)^+.
\]
By the fixed-point identity,
\[
\sum_{e\in U}x_e(v_e-T_e)^+
=
y(U).
\]
It therefore remains to show that
\[
\sum_{e\in U}x_eT_e
\leq
q\,y(U).
\]

Since
\[
T_e=\sum_{i=1}^q p^i_e,
\]
we have
\[
\sum_{e\in U}x_eT_e
=
\sum_{i=1}^q
\sum_{e\in U}x_ep^i_e.
\]
Fix a matroid $M_i$. Let
\[
V^i_1,\ldots,V^i_{m_i}
\]
be its density decomposition, and define
\[
R^i_j:=r_{N^i_j}(V^i_j),
\qquad
X^i_j:=x(V^i_j).
\]
Then
\[
\sum_{e\in U}x_ep^i_e
=
\sum_{j=1}^{m_i}p^i_jX^i_j.
\]

As in the rank-dominance argument from the previous section, for every
$k=1,\ldots,m_i$,
\[
\sum_{j=1}^kX^i_j
\leq
\sum_{j=1}^kR^i_j.
\]
Indeed, if
\[
A^i_k:=V^i_1\cup\cdots\cup V^i_k,
\]
then
\[
\sum_{j=1}^kR^i_j
=
r_{M_i}(A^i_k)
\geq
x(A^i_k)
=
\sum_{j=1}^kX^i_j,
\]
where we use $x\in P(M_i)$.

Moreover, the density prices are non-increasing:
\[
p^i_1\geq p^i_2\geq\cdots\geq p^i_{m_i}\geq 0.
\]
Hence, by summation by parts,
\[
\sum_{j=1}^{m_i}p^i_jX^i_j
\leq
\sum_{j=1}^{m_i}p^i_jR^i_j.
\]
By the definition of the density prices,
\[
p^i_jR^i_j
=
y(V^i_j).
\]
Since the blocks form a partition of $U$,
\[
\sum_{j=1}^{m_i}p^i_jR^i_j
=
\sum_{j=1}^{m_i}y(V^i_j)
=
y(U).
\]
Therefore,
\[
\sum_{e\in U}x_ep^i_e
\leq
y(U).
\]

Summing over all $q$ matroids gives
\[
\sum_{e\in U}x_eT_e
\leq
q\,y(U).
\]
Consequently,
\[
\sum_{e\in U}x_ev_e
\leq
q\,y(U)+y(U)
=
(q+1)y(U),
\]
as claimed.
\end{proof}

\begin{theorem}
\label{thm:q-plus-one}
For the intersection of $q$ matroids, the algorithm above satisfies
\[
\mathbb E[v(I)]
\geq
\frac{1}{q+1}
\sum_{e\in U}x_ev_e.
\]
Consequently, it is a $(q+1)$-approximation with respect to the prophet
benchmark.
\end{theorem}

\begin{proof}
By \Cref{lem:intersection-extracted-feasibility}, every set independent in
each extracted matroid $\widehat M_i$ is independent in the corresponding
original matroid $M_i$. Hence the set accepted by the algorithm is feasible
for the intersection
\[
\bigcap_{i=1}^q \mathcal I_i.
\]

By \Cref{lem:intersection-gain-surplus},
\[
\mathbb E[v(I)]
\geq
y(U).
\]
On the other hand, by \Cref{lem:q-plus-one-accounting},
\[
\sum_{e\in U}x_ev_e
\leq
(q+1)y(U).
\]
Therefore,
\[
y(U)
\geq
\frac{1}{q+1}
\sum_{e\in U}x_ev_e.
\]
Combining \Cref{lem:intersection-gain-surplus,lem:q-plus-one-accounting} gives
\[
\mathbb E[v(I)]
\geq
\frac{1}{q+1}
\sum_{e\in U}x_ev_e.
\]

Finally, the ex-ante relaxation upper-bounds the prophet benchmark. Hence the
algorithm is a $(q+1)$-approximation with respect to the prophet.
\end{proof}

\begin{remark}[Polynomial-time implementation]
By \Cref{lem:matroid-intersection-poly-time}, every density decomposition and
its prices can be computed in polynomial time, and by
\Cref{lem:coupled-fixed-point-polytime}, a surplus vector satisfying the
coupled fixed-point condition can be computed in polynomial time. The
extracted direct-sum matroids and their independence tests have polynomial
oracle overhead. Therefore, once the ex-ante solution $x$ is available, the
entire matroid-intersection procedure of this section runs in polynomial time.
\end{remark}

\section*{Acknowledgements}
The first author acknowledges the use of ChatGPT, specifically GPT-5.5 Thinking, in the development of this work. The system was used to explore proofs of \cref{lem:extraction-charging} for special classes of matroids, and to help organize a clean procedure for decomposing a given instance into nice instances. It was also used to assist with editing and writing parts of the manuscript. The authors take full responsibility for all mathematical claims, proofs, and final editorial decisions.

In the development of \Cref{sec:surplus-algorithm,sec:q-intersection}, the first author also acknowledges the use of ChatGPT, specifically GPT-5.6 Sol. Motivated by the $(q+1)$-approximation of Alon, Pollner, and Weinberg for intersections of partition matroids, the first author introduced the surplus-threshold formulation of \Cref{sec:surplus-algorithm} and used GPT-5.6 Sol to explore and establish analogues of the structural properties developed in the preceding section, including the corresponding extraction procedure and its analysis. The first author then asked GPT-5.6 Sol to investigate whether the extraction framework of \Cref{sec:surplus-algorithm} could be combined with a Brouwer fixed-point argument to obtain a $(q+1)$-approximation for intersections of arbitrary matroids. The guiding idea was to replace the parts appearing in the partition-matroid setting by the extracted matroid pieces. GPT-5.6 Sol assisted in developing the resulting coupled density decomposition, the fixed-point formulation and proof of existence of the common surplus vector, and the proofs underlying the $(q+1)$ guarantee in \Cref{sec:q-intersection}. The authors subsequently checked, revised, and organized these arguments and take full responsibility for all mathematical claims, proofs, and final editorial decisions.
\newpage

\bibliographystyle{alpha}
\bibliography{main}

@incollection{KrengelSucheston,
  author    = {Ulrich Krengel and Louis Sucheston},
  title     = {On Semiamarts, Amarts, and Processes with Finite Value},
  booktitle = {Probability on Banach Spaces},
  editor    = {James Kuelbs},
  series    = {Advances in Probability and Related Topics},
  volume    = {4},
  pages     = {197--266},
  publisher = {Marcel Dekker},
  address   = {New York},
  year      = {1978}
}

@article{SamuelCahn,
  author  = {Ester Samuel-Cahn},
  title   = {Comparison of Threshold Stop Rules and Maximum for Independent Nonnegative Random Variables},
  journal = {The Annals of Probability},
  volume  = {12},
  number  = {4},
  pages   = {1213--1216},
  year    = {1984},
  doi     = {10.1214/aop/1176993150}
}

@article{KleinbergWeinbergJournal,
  author  = {Robert Kleinberg and S. Matthew Weinberg},
  title   = {Matroid Prophet Inequalities and Applications to Multi-Dimensional Mechanism Design},
  journal = {Games and Economic Behavior},
  volume  = {113},
  pages   = {97--115},
  year    = {2019},
  doi     = {10.1016/j.geb.2014.11.002}
}

@inproceedings{Chawla2010,
  author    = {Shuchi Chawla and Jason D. Hartline and David L. Malec and Balasubramanian Sivan},
  title     = {Multi-Parameter Mechanism Design and Sequential Posted Pricing},
  booktitle = {Proceedings of the Forty-Second ACM Symposium on Theory of Computing},
  series    = {STOC '10},
  pages     = {311--320},
  publisher = {Association for Computing Machinery},
  year      = {2010},
  doi       = {10.1145/1806689.1806733}
}

@article{FeldmanSvenssonZenklusen,
  author  = {Moran Feldman and Ola Svensson and Rico Zenklusen},
  title   = {Online Contention Resolution Schemes with Applications to Bayesian Selection Problems},
  journal = {SIAM Journal on Computing},
  volume  = {50},
  number  = {2},
  pages   = {255--300},
  year    = {2021},
  doi     = {10.1137/18M1226130}
}

@article{ChawlaGraphic,
  author  = {Shuchi Chawla and Kira Goldner and Anna R. Karlin and J. Benjamin Miller},
  title   = {Non-Adaptive Matroid Prophet Inequalities},
  journal = {CoRR},
  volume  = {abs/2011.09406},
  year    = {2020},
  eprint  = {2011.09406},
  archivePrefix = {arXiv},
  primaryClass  = {cs.DS}
}

@article{PashkovichSayutina,
  author  = {Kanstantsin Pashkovich and Alice Sayutina},
  title   = {Non-Adaptive Prophet Inequalities for Minor-Closed Classes of Matroids},
  journal = {Discrete Applied Mathematics},
  volume  = {383},
  pages   = {26--43},
  year    = {2026},
  doi     = {10.1016/j.dam.2025.12.001}
}

@inproceedings{LeeSingla2018,
  author    = {Lee, Euiwoong and Singla, Sahil},
  title     = {Optimal Online Contention Resolution Schemes via Ex-Ante Prophet Inequalities},
  booktitle = {26th Annual European Symposium on Algorithms (ESA)},
  series    = {LIPIcs},
  volume    = {112},
  pages     = {57:1--57:14},
  year      = {2018},
  doi       = {10.4230/LIPIcs.ESA.2018.57},
  note      = {Preprint: arXiv:1806.09251}
}

@misc{jiang2026costnonadaptivitymatroidprophet,
  author        = {Tianle Jiang},
  title         = {On the Cost of Non-Adaptivity in Matroid Prophet Inequalities},
  year          = {2026},
  note          = {arXiv preprint arXiv:2607.02766. Available at
                   \url{https://arxiv.org/abs/2607.02766}}
}

@InProceedings{saxena_et_al:LIPIcs.ITCS.2023.95,
  author =	{Saxena, Raghuvansh R. and Velusamy, Santhoshini and Weinberg, S. Matthew},
  title =	{{An Improved Lower Bound for Matroid Intersection Prophet Inequalities}},
  booktitle =	{14th Innovations in Theoretical Computer Science Conference (ITCS 2023)},
  pages =	{95:1--95:20},
  series =	{Leibniz International Proceedings in Informatics (LIPIcs)},
  ISBN =	{978-3-95977-263-1},
  ISSN =	{1868-8969},
  year =	{2023},
  volume =	{251},
  editor =	{Tauman Kalai, Yael},
  publisher =	{Schloss Dagstuhl -- Leibniz-Zentrum f{\"u}r Informatik},
  address =	{Dagstuhl, Germany},
  URL =		{https://drops.dagstuhl.de/entities/document/10.4230/LIPIcs.ITCS.2023.95},
  URN =		{urn:nbn:de:0030-drops-175986},
  doi =		{10.4230/LIPIcs.ITCS.2023.95}
}

@article{10.1007/s10107-023-02027-2,
author = {Correa, Jos{\'e} and Cristi, Andr{\'e}s and Fielbaum, Andr{\'e}s and Pollner, Tristan and Weinberg, S. Matthew},
title = {Optimal item pricing in online combinatorial auctions},
year = {2023},
issue_date = {Jul 2024},
publisher = {Springer-Verlag},
address = {Berlin, Heidelberg},
volume = {206},
number = {1–2},
issn = {0025-5610},
url = {https://doi.org/10.1007/s10107-023-02027-2},
doi = {10.1007/s10107-023-02027-2},
journal = {Math. Program.},
month = oct,
pages = {429–460},
numpages = {32}
}

\newpage
\appendix
\section{Proof of the Ex-Ante Reduction}
\label{app:exante-reduction}

\begin{proof}[Proof of \Cref{lem:ex-ante-bernoulli}]
Consider an arbitrary prophet inequality instance with independent
nonnegative random variables $\{X_e\}_{e\in U}$ and a feasibility system whose
marginal relaxation is $\mathcal P$. Let
\[
x\in\mathcal P
\]
be an optimal solution to the ex-ante relaxation
\[
\max\left\{
\sum_{e\in U}R_e(x_e):
x\in\mathcal P
\right\}.
\]

For every element $e\in U$, let $A_e$ denote the event that $X_e$ lies in its
top $x_e$-quantile, with fractional tie-breaking if necessary, so that
\[
\Pr[A_e]=x_e.
\]
If $x_e>0$, define
\[
v_e:=\mathbb E[X_e\mid A_e],
\]
and if $x_e=0$, set $v_e=0$. By the definition of $R_e$,
\[
x_ev_e=R_e(x_e)
\]
for every $e\in U$. Consequently,
\[
\sum_{e\in U}x_ev_e
=
\sum_{e\in U}R_e(x_e),
\]
which is the value of the ex-ante relaxation.

We now construct the corresponding Bernoulli instance. Element $e$ is active
with probability $x_e$, and, when active, has value $v_e$. We couple this
Bernoulli instance with the original instance by declaring $e$ to be active
if and only if $A_e$ occurs. Since the original random variables are
independent, the events $\{A_e\}_{e\in U}$ are independent.

Fix an arbitrary arrival permutation chosen before the random values are
realized. Run the assumed
Bernoulli-instance algorithm on the original instance using the coupling
above. When element $e$ arrives, the algorithm treats it as active if and only
if $A_e$ occurs, and then makes exactly the same accept/reject decision that it
would make in the Bernoulli instance. Thus, for every realization of the
activity indicators $\{\mathbf 1_{A_e}\}_{e\in U}$ and of the internal
randomness of the algorithm, the accepted set is exactly the same in the two
coupled instances.

It remains to compare the values obtained in the two instances. Let $B_e$
denote the event that element $e$ is accepted. Since an element can be
accepted only when it is active,
\[
B_e\subseteq A_e.
\]
Conditional on $A_e$, the event $B_e$ depends on the arrival order, the
activity indicators of the other elements, the previously accepted set, and
the internal randomness of the algorithm, but it does not depend on the exact
realized value of $X_e$ within its top $x_e$-quantile. Indeed, the translated
algorithm uses $X_e$ only through the indicator $\mathbf 1_{A_e}$, and the
arrival order is chosen independently of all realized values. Since $X_e$ is
independent of the random variables of the other elements, it follows that
\[
\mathbb E[X_e\mid A_e,B_e]
=
\mathbb E[X_e\mid A_e]
=
v_e.
\]
Therefore,
\[
\begin{aligned}
\mathbb E[X_e\mathbf 1_{B_e}]
&=
\Pr[B_e]\,
\mathbb E[X_e\mid B_e] \\
&=
v_e\Pr[B_e].
\end{aligned}
\]
The right-hand side is exactly the expected contribution of $e$ in the
coupled Bernoulli instance. Summing over all elements gives
\[
\mathbb E[\mathrm{ALG}_{\mathrm{original}}]
=
\mathbb E[\mathrm{ALG}_{\mathrm{Bernoulli}}].
\]

By assumption,
\[
\mathbb E[\mathrm{ALG}_{\mathrm{Bernoulli}}]
\geq
c\sum_{e\in U}x_ev_e
=
c\sum_{e\in U}R_e(x_e),
\]
which is $c$ times the value of the ex-ante relaxation.

Since the vector of marginal selection probabilities of every feasible random
set belongs to $\mathcal P$, the ex-ante relaxation upper bounds the prophet
benchmark. Therefore, the resulting online algorithm obtains at least a
$c$ fraction of the prophet benchmark.

Finally, all quantities used in the reduction---namely the vector $x$, the
events defining the top quantiles, the values $v_e$, and any thresholds or
additional feasibility constraints computed from the Bernoulli
instance---are determined without using the realized values. Hence the
reduction preserves pure non-adaptivity and almost non-adaptivity. Moreover,
if the Bernoulli-instance algorithm is order-oblivious, then the resulting
algorithm is order-oblivious as well.
\end{proof}

\section{Proofs of \cref{lem:extraction-polytime,lem:extraction-Surplus-polytime,lem:matroid-intersection-poly-time}}
\label{sec:appendix-polytime}

\begin{lemma}\label{app-lem:extraction-polytime}
Given an independence oracle for \(M\), the extraction procedure can be
implemented in polynomial time.
\end{lemma}
\begin{proof}
We first explain how to check whether a current Bernoulli matroid instance
\((N,x,v)\), with ground set \(W\), is nice, and how to find a set maximizing
the threshold
\[
    T_N(S)=\frac{w(S)}{r_N(S)+x(S)}
    \qquad
    \emptyset\neq S\subseteq W.
\]
Since \(x_e>0\) for every \(e\in W\), the denominator is positive for every
nonempty set \(S\).

For a parameter \(\lambda\ge 0\), there exists a nonempty set \(S\subseteq W\)
with
\[
    T_N(S)>\lambda
\]
if and only if there exists a set \(S\subseteq W\) with
\[
    w(S)-\lambda x(S)-\lambda r_N(S)>0.
\]
Equivalently,
\[
    \min_{S\subseteq W}
    \left\{
        \lambda r_N(S)+\lambda x(S)-w(S)
    \right\}
    <0.
\]
The set function
\[
    f_\lambda(S)
    :=
    \lambda r_N(S)+\lambda x(S)-w(S)
\]
is submodular, since \(r_N\) is submodular and \(x(S)\) and \(w(S)\) are modular
functions. Therefore, for any fixed \(\lambda\), the above minimization problem
can be solved in polynomial time by submodular function minimization.

The rank function \(r_N\) can be evaluated in polynomial time using the
independence oracle: to compute \(r_N(S)\), greedily build a maximal independent
subset of \(S\). Hence \(f_\lambda\) is available through a polynomial-time
value oracle. Standard parametric submodular minimization, or equivalently
binary search over the rational breakpoints together with submodular
minimization, gives a set maximizing \(T_N(S)\) in polynomial time.

Thus, at each iteration of \Cref{alg:extract}, we can check whether the current
instance is nice and, if it is not nice, find a nonempty set of maximum
threshold in polynomial time. Since each non-final iteration extracts a
nonempty proper subset of the current ground set, the number of iterations is at
most \(|U|\). Therefore the entire extraction procedure runs in polynomial
time.
\end{proof}

In this section, we prove the properties of the extraction procedure used in
\Cref{subsec:surplus-extraction}. Throughout this section, let
\[
(M'_1,U_1,T'_1),\ldots,(M'_\ell,U_\ell,T'_\ell)
\]
be the output of the extraction procedure.

\begin{lemma}
\label{app-lem:extraction-Surplus-polytime}
Given an independence oracle for a loopless matroid $M$, the
\textsc{Surplus-Extract} procedure can be implemented in polynomial time.
Moreover, at every iteration, one can find in polynomial time a flat of the
current matroid that maximizes the surplus threshold.
\end{lemma}

\begin{proof}
Let $N=(W,\mathcal J)$ be a loopless matroid appearing at some iteration of
the extraction procedure. For a parameter $\lambda\geq 0$, consider the set
function
\[
f_\lambda(S)
:=
\lambda r_N(S)
-
\sum_{e\in S}x_e(v_e-\lambda)^+,
\qquad
S\subseteq W.
\]
For every fixed $\lambda$, the function
\[
S\longmapsto \lambda r_N(S)
\]
is submodular, while
\[
S\longmapsto
\sum_{e\in S}x_e(v_e-\lambda)^+
\]
is modular. Hence $f_\lambda$ is submodular.

Moreover, for every nonempty set $S\subseteq W$,
\[
T'_N(S)>\lambda
\]
if and only if
\[
\sum_{e\in S}x_e(v_e-\lambda)^+
>
\lambda r_N(S),
\]
or equivalently,
\[
f_\lambda(S)<0.
\]
Therefore,
\[
\max_{\emptyset\neq S\subseteq W}T'_N(S)>\lambda
\]
if and only if
\[
\min_{S\subseteq W}f_\lambda(S)<0.
\]

For every fixed $\lambda$, the latter minimization problem can be solved in
polynomial time by submodular function minimization. Indeed, the rank function
$r_N$ can be evaluated in polynomial time using an independence oracle for
$N$, and hence $f_\lambda$ is available through a polynomial-time value
oracle.

It follows from standard parametric submodular minimization that we can find
in polynomial time the largest value
\[
\lambda^*
=
\max_{\emptyset\neq S\subseteq W}T'_N(S)
\]
together with a nonempty set $S^*\subseteq W$ satisfying
\[
T'_N(S^*)=\lambda^*.
\]
Equivalently, $S^*$ satisfies
\[
\sum_{e\in S^*}x_e(v_e-\lambda^*)^+
=
\lambda^*r_N(S^*).
\]

We now turn $S^*$ into a flat. By
\Cref{lem:surplus-max-flat},
\[
T'_N(S^*)
\leq
T'_N\bigl(\operatorname{cl}_N(S^*)\bigr).
\]
Since $S^*$ already has maximum surplus threshold, equality must hold.
Therefore,
\[
U_i:=\operatorname{cl}_N(S^*)
\]
is a flat that also maximizes the surplus threshold.

The closure of a set can be computed in polynomial time using the independence
oracle. Indeed, for every $e\in W\setminus S^*$, we test whether
\[
r_N(S^*\cup\{e\})=r_N(S^*).
\]
The elements satisfying this equality are precisely those in
$\operatorname{cl}_N(S^*)$.

Thus, at every iteration, a maximum-surplus-threshold flat can be found in
polynomial time.

Finally, each non-final iteration extracts a nonempty set and contracts it, so
at least one element is removed from the current ground set. Hence the number
of iterations is at most $|U|$. Independence and rank queries in every
restriction and contraction can be simulated with polynomial overhead using
the independence oracle for the original matroid $M$.

Therefore, the entire \textsc{Surplus-Extract} procedure can be implemented
in polynomial time.
\end{proof}

\begin{lemma}[Polynomial-time density decomposition]
\label{app-lem:matroid-intersection-poly-time}
Given an independence oracle for $M_i$, the density decomposition of $y$
can be computed in polynomial time.
\end{lemma}
\begin{proof}
At every iteration, we need to find a set maximizing
\[
\frac{y(S)}{r(S)} .
\]

For a parameter $\lambda\geq 0$, there exists a set with density larger than
$\lambda$ if and only if
\[
y(S)-\lambda r(S)>0
\]
for some set $S$.

Equivalently, we need to solve
\[
\min_{S\subseteq U}
\{\lambda r(S)-y(S)\}.
\]

The function
\[
\lambda r(S)-y(S)
\]
is submodular, since the rank function is submodular and $y(S)$ is modular.
Therefore, it can be minimized in polynomial time using submodular function
minimization.

Using parametric submodular minimization, we can find a maximum-density set
in polynomial time. Since every non-final iteration removes at least one
element, there are at most $|U|$ iterations.
\end{proof}

\begin{lemma}[Polynomial-time computation of the coupled surplus vector]
\label{lem:coupled-fixed-point-polytime}
Given independence oracles for $M_1,\ldots,M_q$, a surplus vector $y$
satisfying
\[
y_e
=
x_e
\left(
v_e-\sum_{i=1}^q D_{M_i}(y)_e
\right)^+
\]
can be computed in polynomial time.
\end{lemma}

\begin{proof}
For each $i\in\{1,\ldots,q\}$ and $p\in\mathbb R_+^U$, define
\[
G_i(p)
:=
\frac{1}{2}
\max_{z\in P(M_i)}
\sum_{e\in U}z_ep_e^2.
\]
For fixed $z\in P(M_i)$, the function
\[
p\longmapsto
\frac{1}{2}\sum_{e\in U}z_ep_e^2
\]
is convex. Hence $G_i$, being the pointwise maximum of convex functions, is
convex.

For nonnegative price vectors
\[
p^1,\ldots,p^q\in\mathbb R_+^U,
\]
write
\[
T_e:=\sum_{i=1}^q p^i_e
\]
and consider the convex potential
\[
\Psi(p^1,\ldots,p^q)
:=
\sum_{i=1}^q G_i(p^i)
+
\frac{1}{2}
\sum_{e\in U}
x_e(v_e-T_e)_+^2.
\]

We minimize $\Psi$ over the nonnegative orthant. First observe that there exists an optimal solution satisfying
\[
T_e\leq v_e
\qquad
\text{for every }e\in U.
\]
Indeed, if $T_e>v_e$, choose an index $i$ with $p^i_e>0$ and decrease
$p^i_e$ while maintaining $T_e\geq v_e$. The second term of $\Psi$ remains
unchanged, while $G_i$ cannot increase, since every vector in $P(M_i)$ is
nonnegative. Repeating this operation yields an optimal solution in the
compact set
\[
K
:=
\left\{
(p^1,\ldots,p^q):
p^i\geq0,\quad
\sum_{i=1}^q p^i_e\leq v_e
\text{ for every }e
\right\}.
\]
Thus $\Psi$ attains its minimum. Fix a minimizer
\[
(p^1,\ldots,p^q),
\]
and define
\[
T_e:=\sum_{i=1}^q p^i_e,
\qquad
y_e:=x_e(v_e-T_e)^+.
\]

We next describe the subdifferential of $G_i$. Let
\[
F_i(p)
:=
\arg\max_{z\in P(M_i)}
\sum_{e\in U}z_ep_e^2.
\]
For fixed $z$, the gradient of
\[
\frac{1}{2}\sum_e z_ep_e^2
\]
is $z\odot p$. Hence the standard subdifferential formula for a pointwise
maximum gives
\[
\partial G_i(p)
=
\{z\odot p:z\in F_i(p)\}.
\]

The derivative of
\[
\frac{1}{2}
\sum_e x_e(v_e-T_e)_+^2
\]
with respect to $p^i_e$ is $-y_e$. Therefore, the first-order optimality
conditions imply that there exists
\[
z^i\in F_i(p^i)
\]
such that
\[
y=z^i\odot p^i.
\]
Indeed, if $p^i_e>0$, optimality gives
\[
z^i_ep^i_e=y_e.
\]
If $p^i_e=0$, the boundary optimality condition gives
\[
-y_e\geq0,
\]
and hence $y_e=0$, so the same equality holds.

Thus, for every $i$,
\[
z^i\in P(M_i),
\qquad
z^i\in
\arg\max_{z\in P(M_i)}
\sum_e z_e(p^i_e)^2,
\qquad
y_e=z^i_ep^i_e.
\]

We now relate $p^i$ to the canonical density decomposition of $y$.
Fix $i$ and omit the superscript. Let the distinct positive values of $p$
be
\[
\tau_1>\tau_2>\cdots>\tau_m>0,
\]
and define
\[
A_j:=\{e:p_e\geq\tau_j\},
\qquad
A_0:=\emptyset,
\]
and
\[
V_j:=A_j\setminus A_{j-1}.
\]

Since $z$ maximizes
\[
\sum_e z_ep_e^2
\]
over $P(M)$, the matroid greedy theorem implies
\[
z(A_j)=r_M(A_j)
\]
for every $j$. Consider the residual matroid after contracting $A_{j-1}$. For every set
$S$ contained in the residual ground set,
\[
z(S)
\leq
r_{M/A_{j-1}}(S),
\]
because
\[
z(A_{j-1})=r_M(A_{j-1})
\]
and $z\in P(M)$. Moreover, every residual element satisfies
\[
p_e\leq\tau_j.
\]
Therefore,
\[
y(S)
=
\sum_{e\in S}z_ep_e
\leq
\tau_j z(S)
\leq
\tau_j r_{M/A_{j-1}}(S).
\]
On the other hand,
\[
\begin{aligned}
y(V_j)
&=
\tau_j z(V_j)\\
&=
\tau_j
\bigl(
r_M(A_j)-r_M(A_{j-1})
\bigr)\\
&=
\tau_j r_{M/A_{j-1}}(V_j).
\end{aligned}
\]
Thus $\tau_j$ is the maximum density in this residual matroid.

The inclusion-wise maximal maximum-density set may contain additional
elements from later price levels. However, every such additional element must
have zero $y$-mass. Indeed, equality in
\[
y(S)
\leq
\tau_j z(S)
\leq
\tau_j r_{M/A_{j-1}}(S)
\]
forces every element of positive $z$-mass in a maximum-density extension to
have price exactly $\tau_j$. Hence an element coming from a strictly lower
price level can be added only with zero $z$-mass, and therefore with zero
$y$-mass. Moreover, such additional elements have zero rank increment, so
they merely complete $V_j$ to its closure.

It follows inductively that the canonical density decomposition of $y$ is
obtained from the price-level decomposition above only by moving zero-$y$
elements to earlier blocks. Consequently,
\[
D_M(y)_e\geq p_e
\]
for every $e$, and
\[
D_M(y)_e=p_e
\qquad
\text{whenever }y_e>0.
\]

Applying this conclusion to every matroid $M_i$, let
\[
\widehat T_e
:=
\sum_{i=1}^q D_{M_i}(y)_e.
\]

If $y_e>0$, then
\[
D_{M_i}(y)_e=p^i_e
\]
for every $i$, and hence
\[
\widehat T_e=T_e.
\]
Therefore,
\[
y_e=x_e(v_e-\widehat T_e)^+.
\]

Now suppose $y_e=0$. If $x_e=0$, the desired identity is immediate.
Otherwise,
\[
0
=
x_e(v_e-T_e)^+
\]
implies
\[
T_e\geq v_e.
\]
Since
\[
D_{M_i}(y)_e\geq p^i_e
\]
for every $i$,
\[
\widehat T_e
\geq
T_e
\geq
v_e.
\]
Thus again
\[
x_e(v_e-\widehat T_e)^+=0=y_e.
\]

We have therefore obtained
\[
y_e
=
x_e
\left(
v_e-\sum_{i=1}^qD_{M_i}(y)_e
\right)^+
\]
for every $e\in U$.

It remains to verify polynomial-time computability. Given
$(p^1,\ldots,p^q)$, each value $G_i(p^i)$ can be computed by finding a
maximum-weight independent set of $M_i$ with weights
\[
(p^i_e)^2.
\]
This can be done in polynomial time by the matroid greedy algorithm using the
independence oracle. If $B_i$ is such a maximum-weight independent set, then
\[
\mathbf 1_{B_i}\odot p^i
\]
is a subgradient of $G_i$ at $p^i$. The value and a subgradient of the second
term of $\Psi$ are explicit. Hence $\Psi$ admits a polynomial-time evaluation
and subgradient oracle on the compact set $K$.

Standard polynomial-time convex optimization therefore finds a minimizer of
$\Psi$, and the construction above then produces the required surplus vector
$y$. Finally, once $y$ is known, each density decomposition
$D_{M_i}(y)$ can be computed in polynomial time by the density-decomposition
procedure proved earlier.
\end{proof}

\end{document}